\documentclass[%
 aps,
  pra,
  11pt,
  onecolumn,
  notitlepage,
  superscriptaddress,
  amsmath,amssymb
]{revtex4-2}
\usepackage[dvipsnames]{xcolor}
\usepackage[normalem]{ulem}
\usepackage{graphicx}
\usepackage{dcolumn}
\usepackage{bm}
\usepackage{hyperref}
\hypersetup{colorlinks = true, linkcolor = [rgb]{0.19411,0.51882,0.667058}, urlcolor = [rgb]{0.125490,0.29542,0.1647058}, citecolor = [rgb]{0.75882,0.37411,0.14117}}
\usepackage{cleveref}
\usepackage{tcolorbox}
\usepackage{setspace}
\usepackage{stmaryrd}
\usepackage{amssymb}
\usepackage{braket}

\graphicspath{{../images/}}
\usepackage[toc,page]{appendix}
\usepackage{mathtools}
\usepackage{xcolor}

\newtheorem{theorem}{Theorem}

\newtheorem{corollary}{Corollary}

\newtheorem{definition}{Definition}

\newtheorem{proposition}{Proposition}
\newtheorem{remark}{Remark}

\newenvironment{proof}[1][Proof]{\noindent\textbf{#1.} }{\ \rule{0.5em}{0.5em}}

\usepackage{braket}

\def\>{\rangle}
\def\<{\langle}

\usepackage[bottom]{footmisc}
\allowdisplaybreaks

\begin{document}

\title{Efficiently computable bounds on the energy-constrained quantum reading capacity}

\author{Vishal Singh}
\email{vishalsph04@gmail.com}
\affiliation{Mathematical Quantum Information RIKEN Hakubi Research Team, 
RIKEN Pioneering Research Institute (PRI) and RIKEN Center for Quantum Computing (RQC), Wako, Saitama 351-0198, Japan}

\author{Zixin Huang}
\affiliation{School of Science, STEM College, RMIT University, Melbourne, VIC 3000, Australia}

\author{Mark M. Wilde}
\affiliation{School of Electrical and Computer Engineering, Cornell University, Ithaca, New
York 14850, USA}
\begin{abstract}
    In quantum reading, classical messages are encoded in sequences of quantum channels and recovered by probing the channels and processing their outputs. We study the reading capacity of a finite family of finite-dimensional channels under an average constraint on the probing energy, allowing arbitrary adaptive operations between channel uses. Using an input-dependent chain rule for the Belavkin–Staszewski relative entropy, we derive a converse bound expressed as an optimization involving channel Choi operators and the operator relative entropy. This formulation admits semidefinite approximations and reduces, without an energy constraint, to a channel information radius. To obtain achievable rates, we construct bilinear semidefinite lower approximations to a standard non-adaptive reading bound. Alternating optimization produces feasible probe states and encoding distributions, whose Holevo information gives a directly evaluable achievable rate. For jointly classical–quantum channels, we derive a separate converse based on the Umegaki relative entropy. Without an energy constraint, this converse matches a non-adaptive achievable rate, recovering a recent result of Pascual Abraldes and Winter: non-adaptive protocols suffice to achieve the unconstrained reading capacity of jointly classical--quantum channels.
\end{abstract}

\maketitle

\section{Introduction}

\subsection{Background and motivation}

Quantum information theory has identified information processing tasks in which quantum resources improve performance beyond that attainable with classical strategies. Quantum metrology~\cite{GLM11} and quantum illumination~\cite{Lloyd08} are prominent examples. Quantum reading provides another setting in which the choice of a quantum probe can have an operational advantage~\cite{Pirandola11,PLGMB11}. In this task, an encoder stores a classical message in a sequence of memory cells, each of which implements a quantum channel selected from a prescribed family. A reader retrieves the message by preparing probe systems, allowing them to interact with the memory cells, and processing the resulting outputs. The relevant resources therefore include the probe states, any retained ancillary systems, and the operations used to decode the message.

The original optical formulation of quantum reading demonstrated that nonclassical light can outperform classical mixtures of coherent states when distinguishing memory cells with different reflectivities under a constraint on the mean number of incident photons~\cite{Pirandola11}. Ref.~\cite{PLGMB11} subsequently developed a capacity formulation based on encoding across multiple cells and collective decoding. The quantum reading capacity quantifies the largest asymptotic rate, in bits per cell, at which classical messages can be recovered with vanishing error. This formulation connects the task of channel discrimination with the coding problem of reliably retrieving information from a large memory.

The general reading problem allows the reader to retain quantum information between successive interactions with the memory. Ref.~\cite{DW19} formulated quantum reading protocols that include arbitrary adaptive operations between channel uses and established converse bounds for this setting. This paper also presented bounds for environment-parametrized memory cells and evaluated capacities for several structured families. Other work has shown that entanglement across probe inputs can increase achievable reading rates~\cite{LP17}.  These results illustrate why determining the capacity requires controlling a substantially larger class of strategies than independent probing followed by separate measurements.

When quantum memories are probed by photonic systems, such as the ones considered in~\cite{Pirandola11,PLGMB11}, it becomes necessary to introduce energy constraints on the probe states to ensure the reading rate can be achieved by physical processes. Similar energy constraints were also considered in~\cite{DW19}. In this work, we require the reading protocol to follow an average energy constraint, where the average energy of the probe states per memory cell cannot exceed a fixed energy budget. This allows the maximum flexibility to the decoder, where they can allocate large energy to read some memory cells by reducing the energy allocation for reading other memory cells.


Existing information-theoretic bounds do not, by themselves, provide a general numerical procedure for evaluating the energy-constrained capacity. In particular, the general adaptive converse of~\cite{DW19}  involves an optimization of an information gain over input ensembles and a reference system whose dimension is not bounded in its formulation. Bounds based on particular channel simulations can be effective when a suitable simulation is available, but they do not directly supply a uniform computational prescription for an arbitrary finite channel family. Likewise, a non-adaptive achievable rate can be evaluated for a fixed probe and encoding distribution, while jointly optimizing these choices remains a separate numerical problem. These considerations motivate bounds that combine operational validity with explicit finite-dimensional optimization procedures.

\subsection{Summary of results}

In this work, we derive an upper bound on the energy-constrained reading capacity of an arbitrary finite family of finite-dimensional quantum channels. The bound applies to the general adaptive protocols described above. Our approach uses an input-dependent chain rule for the Belavkin--Staszewski relative entropy~\cite{HLSW26} to bound the information accumulated over the entire protocol. Expressing the resulting quantity through the channel Choi operators converts the average-energy constraint into an operator inequality involving the Hamiltonian. Semidefinite approximations of the operator relative entropy~\cite{FS19} then provide a means of evaluating the bound numerically. In the absence of an energy constraint, the converse reduces to an information radius defined by the Belavkin--Staszewski channel relative entropy.

We also develop a numerical approach to the familiar non-adaptive achievable rate obtained by probing each cell with the same state and jointly decoding the outputs. Using an integral representation of the Umegaki relative entropy~\cite{Jen24} and its semidefinite discretizations~\cite{KS26,KW26}, we construct a family of bilinear semidefinite optimization problems that lower-bound this rate. Each feasible solution provides a probe state and an encoding distribution, whose Holevo information can then be evaluated directly. Alternating optimization offers a practical way to search for such solutions, with each block update given by a semidefinite program. This procedure does not require finding the global optimum in order to produce an achievable rate.

For jointly classical--quantum channels, in which every channel measures its input in the same orthonormal basis, we obtain a converse based directly on the Umegaki relative entropy. Its derivation uses the input-dependent form of the classical--quantum channel chain rule~\cite{WBHK20}. Without an energy constraint, this converse matches a non-adaptive achievable rate and yields an information-radius characterization of the capacity. In particular, we recover a recent result of~\cite{Pascual26,PW26} that the unconstrained reading capacity of a set of jointly classical--quantum channels is achieved by a non-adaptive protocol.

\subsection{Paper organization}

The remainder of this paper is organized as follows: Section~\ref{sec:notations} introduces the notation and relative entropies used throughout the work. Section~\ref{sec:read_cap_defn} defines quantum reading capacity with energy constraints. Section~\ref{sec:read_cap_bnd} establishes the converse and achievable bounds and addresses the task of quantum reading with jointly classical--quantum channels, and Section~\ref{sec:numerics} presents numerical examples.

\section{Notation and Preliminaries}\label{sec:notations}

In this section, we review background material on the three major elements that we use in the rest of the work: quantum states, channels, and divergences.

\subsection{Quantum states and channels}

A quantum state $\rho_A$ is a positive semidefinite, unit-trace operator acting on a Hilbert space~$\mathcal{H}_A$. We denote the set of all linear operators acting on the Hilbert space $\mathcal{H}_A$ by $\mathcal{L}(A)$ and the set of all quantum states acting on this Hilbert space by $\mathcal{S}(A)$. A bipartite quantum state $\rho_{AB}$ acting on the Hilbert space $\mathcal{H}_{A}\otimes \mathcal{H}_B$ is called separable if it can be written as
\begin{equation}
	\rho_{AB} = \sum_{x\in \mathcal{X}} p(x)\sigma^x_{A}\otimes\tau^x_{B},
 \label{eq:def-sep-state}
\end{equation}
where $\{p(x)\}_{x\in \mathcal{X}}$ is a probability distribution and $\{\sigma^x_{A}\}_{x\in \mathcal{X}}$ and $\{\tau^x_{B}\}_{x\in \mathcal{X}}$ are sets of states. Any quantum state that is not separable is said to be entangled. The maximally entangled state vector in the Hilbert space $\mathcal{H}_{A}\otimes\mathcal{H}_B$ is denoted as follows:
\begin{equation}
	|\Phi^d\rangle_{AB} \coloneqq \frac{1}{\sqrt{d}}\sum_{i=0}^{d-1} |i\rangle_A|i\rangle_B,
\end{equation}
where $\{|i\rangle\}_{i=0}^{d-1}$ is an orthonormal basis and $d$ is the Schmidt rank of the state. We denote the corresponding density operator as $\Phi^d_{AB} \equiv |\Phi^d\rangle\!\langle \Phi^d|_{AB}$. 

A quantum channel $\mathcal{N}_{A\to B}$ is a completely positive (CP) and trace-preserving (TP) linear map that takes an operator acting on the Hilbert space $\mathcal{H}_A$ as input and outputs an operator acting on the Hilbert space $\mathcal{H}_B$. We denote the Choi operator of a channel $\mathcal{N}_{A\to B}$ by  $\Gamma^{\mathcal{N}}_{RB}$, which is defined as follows:
\begin{equation}\label{eq:Choi_op_defn}
	\Gamma^{\mathcal{N}}_{RB} \coloneqq \mathcal{N}_{A\to B}\!\left(d\Phi^d_{RA}\right),
\end{equation}
where $\Phi^d_{RA}$ is the maximally entangled state of Schmidt rank $d$ and system $R$ is isomorphic to system $A$. The normalized Choi operator is called the Choi state of the channel, and it is defined as follows:
\begin{equation}
    \Phi^{\mathcal{N}}_{RB} \coloneqq \mathcal{N}_{A\to B}\!\left(\Phi^d_{RA}\right).
\end{equation}

For a bipartite state $\rho_{AB}$, we use the notation $\rho_A \coloneqq \operatorname{Tr}_B\!\left[\rho_{AB}\right]$.

\subsection{Quantum relative entropies and their identities}

The generalized divergence of states and channels plays an important role in our investigation. 

Let $\mathbb{R}$ denote the field of real numbers. A generalized divergence of states~\cite{PV10} is a functional $\mathbf{D}\colon \mathcal{S}(A)\times \mathcal{S}(A) \to \mathbb{R}\cup \{+\infty\}$, such that, for arbitrary states $\rho_A,\sigma_A\in \mathcal{S}(A)$ and an arbitrary channel $\mathcal{N}_{A\to B}$, the data-processing inequality holds
\begin{equation}
    \mathbf{D}\!\left(\rho_A\Vert\sigma_A\right) \ge \mathbf{D}\!\left(\mathcal{N}_{A\to B}(\rho_A)\Vert\mathcal{N}_{A\to B}(\sigma_A)\right).
\end{equation}

Let us briefly discuss some exemplary generalized divergences and their properties that are relevant to this work. 

\begin{itemize}
    \item \textbf{Hypothesis-testing relative entropy:} The hypothesis-testing relative entropy between two states is defined for a parameter $\varepsilon \in [0,1]$ as follows~\cite{BD10,BD11,WR12}:
    \begin{equation}\label{eq:hypo_test_rel_ent_defn}
    D^{\varepsilon}_H\!\left(\rho\Vert\sigma\right) \coloneqq -\log_2 \inf_{0\le \Lambda \le I}\left\{\operatorname{Tr}\!\left[\Lambda\sigma\right]:\operatorname{Tr}\!\left[\Lambda\rho\right]\ge 1-\varepsilon\right\}.
\end{equation}

\item \textbf{Umegaki relative entropy:} The Umegaki relative entropy, also known as quantum relative entropy, is defined as follows~\cite{Ume62}:
\begin{equation}
    D\!\left(\rho\Vert\sigma\right) \coloneqq \Bigg\{\begin{array}{cc}
            \operatorname{Tr}\!\left[\rho\!\left(\log_2\rho - \log_2\sigma\right)\right] & \text{if }\operatorname{supp}(\rho) \subseteq \operatorname{supp}(\sigma)   \\
             +\infty & \text{otherwise} 
        \end{array},
\end{equation}
where the logarithm is taken on the support of its argument.

The Umegaki relative entropy between two states $\rho$ and $\sigma$ is related to the hypothesis-testing relative entropy as follows~\cite{WR12}:
\begin{equation}\label{eq:hypo_test_vs_umegaki}
    D^{\varepsilon}_H\!\left(\rho\middle\Vert\sigma\right) \le \frac{1}{1-\varepsilon}\!\left(D\!\left(\rho\middle\Vert\sigma\right) + h_2(\varepsilon)\right) \qquad \forall \varepsilon\in[0,1),
\end{equation}
where $h_2(\varepsilon)\coloneqq -\varepsilon\log_2(\varepsilon) - (1-\varepsilon)\log_2(1-\varepsilon)$ is the binary entropy.

The Umegaki relative entropy has the direct-sum property~\cite[Exercise 11.8.8]{Wilde17book}. Let $p\colon \mathcal{X} \to [0,1]$ be a probability distribution over a finite-dimensional alphabet $\mathcal{X}$. Let $\left\{\rho^x_A\right\}_{x\in \mathcal{X}}$ and $\left\{\sigma^x_A\right\}_{x\in \mathcal{X}}$ be sets of states. Then,
\begin{equation}\label{eq:Umegaki_direct_sum}
    D\!\left(\sum_{x\in \mathcal{X}}p(x)|x\rangle\!\langle x|_X\otimes \rho^x_A\middle\Vert\sum_{x\in \mathcal{X}}p(x)|x\rangle\!\langle x|_X\otimes \sigma^x_A\right) = \sum_{x\in \mathcal{X}}p(x)D\!\left(\rho^x_A\Vert\sigma^x_A\right).
\end{equation}

\item \textbf{Belavkin--Staszewski relative entropy:} The Belavkin--Staszewski relative entropy of states is defined as follows~\cite{BS82}:
\begin{equation}\label{eq:BS_defn}
        \widehat{D}\!\left(\rho\Vert\sigma\right) \coloneqq \Bigg\{\begin{array}{cc}
            \operatorname{Tr}\!\left[\rho\log_2\!\left(\sqrt{\rho}\sigma^{-1}\sqrt{\rho}\right)\right] & \text{if }\operatorname{supp}(\rho) \subseteq \operatorname{supp}(\sigma)   \\
             +\infty & \text{otherwise} 
        \end{array},
    \end{equation}
    where $\sigma^{-1}$ is taken on the support of $\sigma$ and the logarithm is evaluated on the support of $\rho$.

    The Belavkin--Staszewski relative entropy of a pair of states $\rho$ and $\sigma$ is bounded from below by the Umegaki relative entropy \cite{HiaiPetz1993GoldenThompson}. That is,
    \begin{equation}\label{eq:BS_ge_Umegaki}
        \widehat{D}(\rho\Vert\sigma) \ge D(\rho\Vert\sigma).
    \end{equation}

\end{itemize}

The generalized divergence of channels is defined in terms of the generalized divergence of states as follows \cite{LKDW18}:
\begin{equation}
    \mathbf{D}\!\left(\mathcal{N}_{A\to B}\middle\Vert\mathcal{M}_{A\to B}\right) \coloneqq \sup_{\rho_{RA}}\mathbf{D}\!\left(\mathcal{N}_{A\to B}\!\left(\rho_{RA}\right)\middle\Vert\mathcal{M}_{A\to B}\!\left(\rho_{RA}\right)\right),
    \label{eq:def-gen-div}
\end{equation}
where the supremum is taken over every state $\rho_{RA}$ with no restriction on the dimension of $R$.

The Belavkin--Staszewski relative entropy of channels is defined as an instance of \eqref{eq:def-gen-div} as follows:
\begin{equation}
    \widehat{D}\!\left(\mathcal{N}_{A\to B}\middle\Vert\mathcal{M}_{A\to B}\right) \coloneqq \sup_{\rho_{RA}}\widehat{D}\!\left(\mathcal{N}_{A\to B}\!\left(\rho_{RA}\right)\middle\Vert\mathcal{M}_{A\to B}\!\left(\rho_{RA}\right)\right),
\end{equation}
has the following chain-rule property~\cite{FF21}:
\begin{equation}\label{eq:geo_ent_amort}
    \widehat{D}\!\left(\mathcal{N}\!\left(\rho\right)\middle\Vert\mathcal{M}\!\left(\sigma\right)\right) \le \widehat{D}\!\left(\mathcal{N}\middle\Vert\mathcal{M}\right) + \widehat{D}\!\left(\rho\Vert\sigma\right).
\end{equation}


A convenient way to study the Belavkin--Staszewski relative entropy of channels is through the operator relative entropy \cite{Furuta1990Operator}, which is defined as follows:
\begin{equation}\label{eq:D_op_defn}
    \widehat{D}_{\operatorname{op}}\!\left(X,Y\right) \coloneqq X^{\frac{1}{2}}\log_2\!\left(X^{\frac{1}{2}}Y^{-1}X^{\frac{1}{2}}\right)X^{\frac{1}{2}} \qquad \operatorname{if }\operatorname{supp}\!\left(X\right) \subseteq \operatorname{supp}(Y),
\end{equation}
where the inverse is taken on the support of $Y$.
It was shown in~\cite[Theorem~2]{FF21} that the Belavkin--Staszewski relative entropy of channels can be expressed in the following closed form:
\begin{equation}\label{eq:BS_ch_closed_form}
    \widehat{D}\!\left(\mathcal{N}_{A\to B}\middle\Vert\mathcal{M}_{A\to B}\right) = \left\Vert\operatorname{Tr}_B\!\left[\widehat{D}_{\operatorname{op}}\!\left(\Gamma^{\mathcal{N}}_{AB},\Gamma^{\mathcal{M}}_{AB}\right)\right]\right\Vert_{\infty},
\end{equation}
where $\Gamma^{\mathcal{N}}$ and $\Gamma^{\mathcal{M}}$ are the Choi operators of the channels $\mathcal{N}$ and $\mathcal{M}$, respectively.

The following input-dependent chain rule for the Belavkin--Staszewski relative entropy in terms of the operator relative entropy was obtained in~\cite[Appendix B]{HLSW26}: 
\begin{equation}\label{eq:ec_chain_rule_ing}
    \widehat{D}\!\left(\mathcal{N}_{A\to B}\!\left(\rho_{RA}\right)\Vert\mathcal{M}_{A\to B}\!\left(\sigma_{RA}\right)\right) \le \widehat{D}\!\left(\rho_{RA}\Vert\sigma_{RA}\right) + \operatorname{Tr}\!\left[\left(\rho^T_A\otimes I_B\right)\widehat{D}_{\operatorname{op}}\!\left(\Gamma^{\mathcal{N}}_{AB},\Gamma^{\mathcal{M}}_{AB}\right)\right]
\end{equation}
if $\operatorname{supp}\!\left(\Gamma^{\mathcal{N}}_{AB}\right) \subseteq \operatorname{supp}\!\left(\Gamma^{\mathcal{M}}_{AB}\right)$, where $\rho^T_A$ denotes the transpose of $\rho_A$ with respect to the chosen basis for system $A$ in the representation of the Choi operators $\Gamma^{\mathcal{N}}$ and $\Gamma^{\mathcal{M}}$.

\section{Reading capacity of channels}\label{sec:read_cap_defn}

In this section, we describe the task of quantum reading and define the energy-constrained reading capacity of a set of channels, which is the quantity we investigate in the remainder of this paper. We follow~\cite{DW19} for our definition of the energy-constrained reading capacity.

In a quantum reading protocol, an encoder has the ability to inscribe a memory cell with a quantum channel selected from a set of channels $\left\{\mathcal{N}^{x}_{A\to B}\right\}_{x\in \mathcal{X}}$. The encoder encodes a message $m\in \mathcal{M}$ in the quantum memory using a codeword $x^n(m)\coloneqq [x_1(m),x_2(m),\ldots,x_n(m)]$ such that the $i^{\text{th}}$ memory cell is inscribed by the channel $\mathcal{N}^{x_i(m)}_{A\to B}$. The decoder then uses a probe state $\rho_{AR_1}$, a set of intermediate processing channels $\left\{\mathcal{A}^i_{BR_i\to AR_{i+1}}\right\}_{i\in [n-1]}$, and a decoding POVM $\left\{\Lambda^{\hat{m}}_{BR_n}\right\}_{\hat{m}\in \mathcal{M}}$ to retrieve the original message $m$, where systems $R_1, R_2,\ldots, R_n$ are of arbitrary dimension. 

Let us denote the set of channels, $\left\{\mathcal{N}^{x}_{A\to B}\right\}_{x\in \mathcal{X}}$, by the symbol $\mathcal{N}$. Let us also define the following set that describes a reading protocol:
\begin{equation}
    \mathcal{P}^{(n)}\coloneqq \left\{\rho_{AR_1},\left\{\mathcal{A}^i_{BR_i\to AR_{i+1}}\right\}_{i\in [n-1]},\left\{\Lambda^{\hat{m}}_{BR_n}\right\}_{\hat{m}\in \mathcal{M}}\right\}.
\end{equation}
The error in the decoding an arbitrary message $m$ encoded in $n$ channels using a protocol $\mathcal{P}^{(n)}$ is defined as follows:
\begin{equation}
    p_{\operatorname{err}}\!\left(\mathcal{N},\mathcal{P}^{(n)}\right) \coloneqq \max_{m\in \mathcal{M}}\sum_{\hat{m}\in \mathcal{M}\setminus \{m\}}\operatorname{Tr}\!\left[\Lambda^{\hat{m}}_{BR_n}\left(\mathcal{N}^{x_n(m)}\circ\mathcal{A}^{n-1}\circ\mathcal{N}^{x_{n-1}(m)}\circ\cdots\circ\mathcal{N}^{x_1(m)}\!\left(\rho_{AR_1}\right)\right)\right].
\end{equation}

\begin{definition}
    The $n$-shot reading capacity of a set of channels $\mathcal{N}\coloneqq \left\{\mathcal{N}^{x}_{A\to B}\right\}_{x\in \mathcal{X}}$ is defined as follows:
    \begin{equation}
        R^{\varepsilon,n}\!\left(\mathcal{N}\right) \coloneqq \sup_{\mathcal{P}^{(n)}}\left\{
             \log_2|\mathcal{M}|:  
              p_{\operatorname{err}}\!\left(\mathcal{N},\mathcal{P}^{(n)}\right) \le \varepsilon
        \right\}, 
    \end{equation}
    where 
    \begin{equation}
    \mathcal{P}^{(n)}\coloneqq \left\{\rho_{AR_1},\left\{\mathcal{A}^i_{BR_i\to AR_{i+1}}\right\}_{i\in [n-1]},\left\{\Lambda^{\hat{m}}_{BR_n}\right\}_{\hat{m}\in \mathcal{M}}\right\}. 
    \label{eq:protocol-P}
    \end{equation}
\end{definition}

\begin{definition}
    The reading capacity of a set of channels $\mathcal{N}$ is defined as follows:
    \begin{equation}
        R\!\left(\mathcal{N}\right) \coloneqq \liminf_{\varepsilon \to 0^+}\liminf_{n\to \infty}\frac{1}{n}R^{\varepsilon, n}\!\left(\mathcal{N}\right).
    \end{equation}
\end{definition}


In practical scenarios, the decoder does not have access to arbitrary probe states and post-processing channels. Instead, the decoder is restricted by some energy constraint. In light of this practical restriction, we define the energy-constrained $n$-shot reading capacity of a set of channels, where the average energy of the states used to probe each channel cannot exceed a given threshold. Consider a reading protocol $\mathcal{P}^{(n)}$ of the form in \eqref{eq:protocol-P}. The state that probes the $i^{\operatorname{th}}$ channel in this reading protocol is given by the following:
\begin{equation}
    \rho^{i-1,m}_{AR_{i}} \coloneqq  \mathcal{A}^{i-1}_{BR_{i-1}\to AR_{i}}\circ\mathcal{N}^{x_{i-1}(m)}_{A\to B}\!\left(\rho^{i-2,m}_{AR_{i-1}}\right) \qquad \forall i\in \left\{2,3,\ldots,n\right\}, m\in \mathcal{M}
\end{equation}
where $\rho^{0,m}_{AR_1} = \rho_{AR_1}$. Let $H_A$ be the Hamiltonian associated with the probing system, and let $E$ be the average-energy threshold. Then the $n$-shot reading capacity of a set of channels $\mathcal{N}$ is defined as follows:
\begin{equation}\label{eq:n_shot_ec_read_cap_defn}
        R^{\varepsilon,n}_{H,E}\!\left(\mathcal{N}\right) \coloneqq \sup_{\mathcal{P}^{(n)}}\left\{
             \log_2|\mathcal{M}|:  
              p_{\operatorname{err}}\!\left(\mathcal{N},\mathcal{P}^{(n)}\right) \le \varepsilon, \operatorname{Tr}\!\left[H_A\overline{\rho}^m_A\right]\le E \quad \forall m\in\mathcal{M}
        \right\}, 
\end{equation}
where 
\begin{equation}\label{eq:avg_state_defn}
    \overline{\rho}^m_A \coloneqq \frac{1}{n}\sum_{i=0}^{n-1}\operatorname{Tr}_{R_{i+1}}\!\left[\rho^{i,m}_{AR_{i+1}}\right] \qquad \forall m\in \mathcal{M}.
\end{equation}
The energy-constrained reading capacity of the set of channels $\mathcal{N}$, for a given Hamiltonian $H$ and average-energy threshold $E$, is then  defined as follows:
\begin{equation}\label{eq:ec_read_cap_defn}
        R_{H,E}\!\left(\mathcal{N}\right) \coloneqq \liminf_{\varepsilon \to 0^+}\liminf_{n\to \infty}\frac{1}{n}R^{\varepsilon, n}_{H,E}\!\left(\mathcal{N}\right).
\end{equation}

\section{Efficiently computable bounds on the energy-constrained reading capacity}\label{sec:read_cap_bnd}

In this section, we derive efficiently computable upper and lower bounds on the energy-constrained reading capacity of a set of channels. Here, we use the term efficiently computable to indicate that the bounds can be computed in $O(\text{poly}(d))$, where $d$ is the dimension of the Choi operators of the channels.

\subsection{Upper bound on the energy-constrained reading capacity}

Let us first consider an upper bound on the energy-constrained reading capacity from above. Our first main result is encapsulated in Theorem~\ref{theo:ec_read_cap_ub} stated below.

\begin{theorem}\label{theo:ec_read_cap_ub}
    The reading capacity of a set of channels $\mathcal{N}\coloneqq \left\{\mathcal{N}^x_{A\to B}\right\}_{x\in \mathcal{X}}$ is bounded from above as follows:
    \begin{equation}\label{eq:ec_read_cap_ub}
    R_{H,E}\!\left(\mathcal{N}\right) \le \inf_{\Gamma^{\mathcal{S}}_{AB}\ge 0}\inf_{\lambda \in \mathbb{R},\mu\ge 0}\left\{
    \begin{array}{c}
         \lambda + \mu E: \\
         \left(\operatorname{Tr}_B\!\left[\widehat{D}_{\operatorname{op}}\!\left(\Gamma^{x}_{AB},\Gamma^{\mathcal{S}}_{AB}\right)\right]\right)^T \le \lambda I + \mu H \qquad \forall x\in \mathcal{X},\\
         \operatorname{Tr}_B\!\left[\Gamma^{\mathcal{S}}_{AB}\right] = I_A
    \end{array}
    \right\},
\end{equation}
where $\Gamma^x_{AB}$ is the Choi operator of $\mathcal{N}^x_{A\to B}$ for every $x\in \mathcal{X}$ and $\widehat{D}_{\operatorname{op}}$ is defined in~\eqref{eq:D_op_defn}.
\end{theorem}
\begin{proof}
    Let $\left\{\mathcal{N}^{x}_{A\to B}\right\}_{x\in \mathcal{X}}$ be a set of channels, and let $\left\{x^n(m)\right\}_{m\in \mathcal{M}_n}$ be a codebook that is being used to encode a message $m$ into $n$ channels. Consider an arbitrary quantum reading protocol with probe state $\rho_{AR_1}$, intermediate processing channels $\left\{\mathcal{A}^i_{BR_i\to AR_{i+1}}\right\}_{i\in [n-1]}$, and POVM $\left\{\Lambda^{\hat{m}}_{BR_n}\right\}_{\hat{m}\in \mathcal{M}_n}$ such that the average-energy constraint is satisfied  with respect to the Hamiltonian~$H$ and energy $E$ (see~\eqref{eq:n_shot_ec_read_cap_defn} and~\eqref{eq:avg_state_defn}). We will use the following notations throughout the proof for simplicity:
    \begin{align}
        \mathcal{N} &\coloneqq \left\{\mathcal{N}^x_{A\to B}\right\}_{x\in \mathcal{X}},\\
        \mathcal{P}^{(n)} &\coloneqq \left\{\rho_{AR_1},\left\{\mathcal{A}^i_{BR_i\to AR_{i+1}}\right\}_{i\in [n-1]}, \left\{\Lambda^{\hat{m}}_{BR_n}\right\}_{\hat{m}\in \mathcal{M}_n}\right\}.
    \end{align}

    Let us define the following measurement channel:
    \begin{equation}
        \mathcal{L}_{BR_n\to \hat{M}}(\cdot) \coloneqq \sum_{\hat{m}\in \mathcal{M}_n}\operatorname{Tr}\!\left[\Lambda^{\hat{m}}_{BR_n}(\cdot)\right]|\hat{m}\rangle\!\langle \hat{m}|_{\hat{M}}.
    \end{equation}
    The worst-case error of this reading protocol is then given by the following:
    \begin{equation}
        p_{\operatorname{err}}\!\left(\mathcal{N}, \mathcal{P}^{(n)}\right) = \max_{m\in \mathcal{M}_n}\sum_{\hat{m}\in \mathcal{M}_n\setminus \{m\}}\langle \hat{m}|\mathcal{L}\circ\mathcal{N}^{x_n(m)}\circ\mathcal{A}^{n-1}\circ\cdots\circ\mathcal{N}^{x_1(m)}\!\left(\rho_{AR_1}\right)|\hat{m}\rangle,
    \end{equation}
    and the average case error of this reading protocol is given by the following:
    \begin{equation}\label{eq:avg_case_err_exp}
        \overline{p}_{\operatorname{err}}\!\left(\mathcal{N},\mathcal{P}^{(n)}\right) = \frac{1}{|\mathcal{M}_n|}\sum_{\substack{m,\hat{m}\in \mathcal{M}_n,\\ \hat{m}\neq m}}\langle \hat{m}|\mathcal{L}\circ\mathcal{N}^{x_n(m)}\circ\mathcal{A}^{n-1}\circ\cdots\circ\mathcal{N}^{x_1(m)}\!\left(\rho_{AR_1}\right)|\hat{m}\rangle.
    \end{equation}
    Now consider the following channel:
    \begin{equation}
        \mathcal{Q}^{\mathcal{N}}_{M\to \hat{M}}(\cdot) \coloneqq \sum_{m\in \mathcal{M}_n}\mathcal{L}\circ\mathcal{N}^{x_n(m)}\circ\mathcal{A}^{n-1}\circ\cdots\circ\mathcal{N}^{x_1(m)}\!\left(\rho_{AR_1}\right)\langle m|(\cdot)|m\rangle .
    \end{equation}
    Let $\omega_{MM'}$ denote the maximally-classically correlated state, and let $\Pi_{MM'}$ denote the test for classical correlations. That is,
    \begin{align}
        \omega_{MM'} &\coloneqq \frac{1}{|\mathcal{M}_n|}\sum_{m\in \mathcal{M}_n}|m\rangle\!\langle m|_{M}\otimes |m\rangle\!\langle m|_{M'},\\
        \Pi_{MM'} &\coloneqq \sum_{m\in \mathcal{M}_n}|m\rangle\!\langle m|_{M}\otimes |m\rangle\!\langle m|_{M'}.
    \end{align}
    Let us consider the following expression:
    \begin{align}
        &\operatorname{Tr}\!\left[\Pi_{M'\hat{M}}\mathcal{Q}^{\mathcal{N}}_{M\to \hat{M}}\!\left(\omega_{MM'}\right)\right]\notag\\ &= \frac{1}{|\mathcal{M}_n|}\sum_{m\in \mathcal{M}_n}\operatorname{Tr}\!\left[|m\rangle\!\langle m|_{M'}\langle m|_{\hat{M}}\mathcal{L}\circ\mathcal{N}^{x_n(m)}\circ\mathcal{A}^{n-1}\circ\cdots\circ\mathcal{N}^{x_1(m)}\!\left(\rho_{AR_1}\right)|m\rangle_{\hat{M}}\right]\\
        &= \frac{1}{|\mathcal{M}_n|}\sum_{m\in \mathcal{M}_n}\langle m|_{\hat{M}}\mathcal{L}\circ\mathcal{N}^{x_n(m)}\circ\mathcal{A}^{n-1}\circ\cdots\circ\mathcal{N}^{x_1(m)}\!\left(\rho_{AR_1}\right)|m\rangle_{\hat{M}}\label{eq:avg_succ_prob_proj}\\
        &= 1-\overline{p}_{\operatorname{err}}\!\left(\mathcal{N},\mathcal{P}^{(n)}\right),\label{eq:avg_succ_prob_err_inv}
    \end{align}
    where the first equality follows from simply substituting the expressions for $\Pi_{M'\hat{M}}$, $\omega_{MM'}$, and $\mathcal{Q}^{\mathcal{N}}_{M\to \hat{M}}$, and the last equality follows from~\eqref{eq:avg_case_err_exp}. 

    If the set $\mathcal{N}$ contains only one channel, then it is clearly useless for the task of reading. Let us define
    \begin{equation}
        \mathcal{Q}^{\mathcal{N}_S}_{M\to \hat{M}}(\cdot) \coloneqq \sum_{m\in \mathcal{M}_n}\mathcal{L}\circ\mathcal{S}\circ\mathcal{A}^{n-1}\circ\cdots\circ\mathcal{S}\!\left(\rho_{AR_1}\right)\langle m|(\cdot)|m\rangle,
    \end{equation}
    where $\mathcal{S}_{A\to B}$ is a fixed channel such that $\operatorname{supp}(\Gamma^{\mathcal{N}^x}_{AB})\subseteq \operatorname{supp}(\Gamma^{\mathcal{S}}_{AB})$ for every $x\in \mathcal{X}$. Then,
   \begin{align}
        \operatorname{Tr}\!\left[\Pi_{M'\hat{M}}\mathcal{Q}^{\mathcal{N}_S}_{M\to \hat{M}}\!\left(\omega_{MM'}\right)\right]
        &= \frac{1}{|\mathcal{M}_n|}\sum_{m\in \mathcal{M}_n}\langle m|_{\hat{M}}\mathcal{L}\circ\mathcal{S}\circ\mathcal{A}^{n-1}\circ\cdots\mathcal{A}^{1}\circ\mathcal{S}\!\left(\rho_{AR_1}\right)|m\rangle_{\hat{M}}\\
        &= \frac{1}{|\mathcal{M}_n|}\operatorname{Tr}\!\left[\mathcal{L}\circ\mathcal{S}\circ\mathcal{A}^{n-1}\circ\cdots\mathcal{A}^{1}\circ\mathcal{S}\!\left(\rho_{AR_1}\right)\right]\\
        &= \frac{1}{|\mathcal{M}_n|},\label{eq:corr_test_pass_prob_useless}
    \end{align}
    where the first equality follows from~\eqref{eq:avg_succ_prob_proj}, the second and third equalities follow from the fact that $\mathcal{L}\circ\mathcal{S}\circ\mathcal{A}^{n-1}\circ\cdots\mathcal{A}^{1}\circ\mathcal{S}\!\left(\rho_{AR_1}\right)$ is a quantum state that is independent of the label $m$.

    The equality in~\eqref{eq:avg_succ_prob_err_inv} implies that for every set of channels $\mathcal{N}$ and reading protocol $\mathcal{P}^{(n)}$, if the worst case error in reading is less than or equal to $\varepsilon$ for some $\varepsilon\in (0,1)$, then
    \begin{equation}\label{eq:corr_test_pass_prob_res}
        \operatorname{Tr}\!\left[\Pi_{M'\hat{M}}\mathcal{Q}^{\mathcal{N}}_{M\to \hat{M}}\!\left(\omega_{MM'}\right)\right] \ge 1-\varepsilon
    \end{equation}
    since $\overline{p}_{\operatorname{err}}(\mathcal{N},\mathcal{P}^{(n)}) \le p_{\operatorname{err}}(\mathcal{N},\mathcal{P}^{(n)})$. Recall the definition of hypothesis testing relative entropy from~\eqref{eq:hypo_test_rel_ent_defn}. Since $0\le \Pi_{M'\hat{M}}\le I_{M'\hat{M}}$, \eqref{eq:corr_test_pass_prob_useless} and~\eqref{eq:corr_test_pass_prob_res} imply that
    \begin{equation}\label{eq:n_mess_hypo_test_ub}
        D^{\varepsilon}_H\!\left(\mathcal{Q}^{\mathcal{N}}_{M\to \hat{M}}\!\left(\omega_{MM'}\right)\middle\Vert\mathcal{Q}^{\mathcal{N}_S}_{M\to \hat{M}}\!\left(\omega_{MM'}\right)\right) \ge -\log_2\frac{1}{|\mathcal{M}_n|} = \log_2|\mathcal{M}_n|.
    \end{equation}
    Let us define the following states:
    \begin{align}
        \sigma^{(m)}_{\hat{M}}&\coloneqq \mathcal{L}\circ\mathcal{N}^{x_n(m)}\circ\mathcal{A}^{n-1}\circ\cdots\circ\mathcal{N}^{x_1(m)}\!\left(\rho_{AR_1}\right)\\
        \tau_{\hat{M}}&\coloneqq \mathcal{L}\circ\mathcal{S}\circ\mathcal{A}^{n-1}\circ\cdots\circ\mathcal{S}\!\left(\rho_{AR_1}\right).
    \end{align}
    Then,
    \begin{align}
        \mathcal{Q}^{\mathcal{N}}_{M\to \hat{M}}\!\left(\omega_{MM'}\right) &= \frac{1}{|\mathcal{M}_n|}\sum_{m\in \mathcal{M}_n}|m\rangle\!\langle m|_{M'}\otimes\sigma^{(m)}_{\hat{M}}\\
        \mathcal{Q}^{\mathcal{N}_S}_{M\to \hat{M}}\!\left(\omega_{MM'}\right) &= \frac{1}{|\mathcal{M}_n|}\sum_{m\in \mathcal{M}_n}|m\rangle\!\langle m|_{M'}\otimes\tau_{\hat{M}}
    \end{align}
    The inequality in~\eqref{eq:n_mess_hypo_test_ub} then leads to the following:
    \begin{align}
        &\log_2|\mathcal{M}_n|\notag\\
        &\le D^{\varepsilon}_H\!\left(\frac{1}{|\mathcal{M}_n|}\sum_{m\in \mathcal{M}_n}|m\rangle\!\langle m|_{M'}\otimes\sigma^{(m)}_{\hat{M}}\middle\Vert\frac{1}{|\mathcal{M}_n|}\sum_{m\in \mathcal{M}_n}|m\rangle\!\langle m|_{M'}\otimes\tau_{\hat{M}}\right)\\
        &\le \frac{1}{1-\varepsilon}\!\left(D\!\left(\frac{1}{|\mathcal{M}_n|}\sum_{m\in \mathcal{M}_n}|m\rangle\!\langle m|_{M'}\otimes\sigma^{(m)}_{\hat{M}}\middle\Vert\frac{1}{|\mathcal{M}_n|}\sum_{m\in \mathcal{M}_n}|m\rangle\!\langle m|_{M'}\otimes\tau_{\hat{M}}\right) + h_2(\varepsilon)\right)\\
        &= \frac{1}{1-\varepsilon}\!\left(\frac{1}{\left|\mathcal{M}_n\right|}\sum_{m\in \mathcal{M}_n}D\!\left(\sigma^{(m)}_{\hat{M}}\middle\Vert\tau_{\hat{M}}\right) + h_2(\varepsilon)\right)\\
        &\le \max_{m\in \mathcal{M}_n}\frac{1}{1-\varepsilon}\!\left(D\!\left(\sigma^{(m)}_{\hat{M}}\middle\Vert\tau_{\hat{M}}\right) + h_2(\varepsilon)\right),\label{eq:n_mess_umegaki_ub}
    \end{align}
where the second inequality follows from~\eqref{eq:hypo_test_vs_umegaki}, the equality follows from~\eqref{eq:Umegaki_direct_sum}, and the final inequality follows because the average is less than or equal to the maximum of a set of numbers. Since the second term in~\eqref{eq:n_mess_umegaki_ub} is independent of the number of channel uses $n$, the following inequality holds for every $\varepsilon\in [0,1)$:
\begin{equation}\label{eq:read_rate_ub_Umegaki}
    \liminf_{n\to \infty}\frac{\log_2|\mathcal{M}_n|}{n} \le \liminf_{n\to \infty}\max_{m\in \mathcal{M}_n}\frac{1}{n(1-\varepsilon)}D\!\left(\sigma^{(m)}_{\hat{M}}\middle\Vert\tau_{\hat{M}}\right).
\end{equation}
Note that $\sigma^{(m)}_{\hat{M}}$ and $\tau_{\hat{M}}$ implicitly depend on $n$.

Let us define the following states:
\begin{align}
    \sigma^{i,m}_{AR_{i+1}} &\coloneqq \mathcal{A}^i_{BR_{i}\to AR_{i+1}}\circ\mathcal{N}^{x_i(m)}_{A\to B}\!\left(\sigma^{i-1,m}_{AR_i}\right) \qquad \forall i\in [n-1],\\
    \tau^i_{AR_{i+1}} &\coloneqq \mathcal{A}^i_{BR_{i}\to AR_{i+1}}\circ\mathcal{S}_{A\to B}\!\left(\tau^{i-1}_{AR_i}\right) \qquad \forall i\in [n-1],\\
    \sigma^{n,m}_{BR_n} &\coloneqq \mathcal{N}^{x_n(m)}_{A\to B}\!\left(\sigma^{n-1,m}_{AR_n}\right),\\
    \tau^{n}_{BR_n} &\coloneqq \mathcal{S}_{A\to B}\!\left(\tau^{n-1}_{AR_n}\right),
\end{align}
where
\begin{equation}
    \sigma^{0,m}_{AR_1} = \tau^0_{AR_1} = \rho_{AR_1}.
\end{equation}
Let us denote the Choi operator of $\mathcal{N}^{x_n(m)}$ by $\Gamma^{x_n(m)}_{AB}$ and the Choi operator of $\mathcal{S}_{A\to B}$ by $\Gamma^{\mathcal{S}}_{AB}$. Then,
\begin{align}
    D\!\left(\sigma^{(m)}_{\hat{M}}\middle\Vert\tau_{\hat{M}}\right) &= D\!\left(\mathcal{L}_{BR_n\to \hat{M}}\!\left(\sigma^{n,m}_{BR_n}\right)\middle\Vert\mathcal{L}_{BR_n\to \hat{M}}\!\left(\tau^{n}_{BR_n}\right)\right)\\
    &\le D\!\left(\sigma^{n,m}_{BR_n}\middle\Vert\tau^{n}_{BR_n}\right)\label{eq:Before_umegaki_to_BS}\\
    &\le \widehat{D}\!\left(\sigma^{n,m}_{BR_n}\middle\Vert\tau^{n}_{BR_n}\right)\label{eq:Umegaki_to_BS}\\
    &\le \widehat{D}\!\left(\sigma^{n-1,m}_{AR_n}\middle\Vert\tau^{n-1}_{AR_n}\right) + \operatorname{Tr}\!\left[\left(\operatorname{Tr}_{R_n}\!\left[\sigma^{n-1,m}_{AR_n}\right]\right)^T\widehat{D}_{\operatorname{op}}\!\left(\Gamma^{x_n(m)}_{AB},\Gamma^{\mathcal{S}}_{AB}\right)\right],\label{eq:ec_chain_rule_apply_1_step}
\end{align}
where the first inequality follows from the data-processing inequality for the Umegaki relative entropy, the second inequality follows from~\eqref{eq:BS_ge_Umegaki}, and the last inequality follows from~\eqref{eq:ec_chain_rule_ing}. We can apply~\eqref{eq:ec_chain_rule_ing} to the first term on the right-hand side of~\eqref{eq:ec_chain_rule_apply_1_step} as well. Following this process sequentially and using the data-processing inequality, we arrive at the following:
\begin{equation}\label{eq:BS_rel_ent_le_avg_BS}
    \widehat{D}\!\left(\sigma^{(m)}_{\hat{M}}\middle\Vert\tau_{\hat{M}}\right) \le \sum_{i=0}^{n-1}\operatorname{Tr}\!\left[\left(\operatorname{Tr}_{R_{i+1}}\!\left[\sigma^{i,m}_{AR_{i+1}}\right]\right)^T\widehat{D}_{\operatorname{op}}\!\left(\Gamma^{x_{i+1}(m)}_{AB},\Gamma^{\mathcal{S}}_{AB}\right)\right] 
\end{equation}
Let $Q_A$ be an operator such that
\begin{equation}\label{eq:ub_BS_op_marg}
    Q_A \ge \operatorname{Tr}_B\!\left[\widehat{D}_{\operatorname{op}}\!\left(\Gamma^{x}_{AB},\Gamma^{\mathcal{S}}_{AB}\right)\right] \qquad \forall x\in \mathcal{X}.
\end{equation}
Then,
\begin{align}
    \widehat{D}\!\left(\sigma^{(m)}_{\hat{M}}\middle\Vert\tau_{\hat{M}}\right) &\le n\operatorname{Tr}\!\left[\left(\overline{\sigma}^{(m)}_A\right)^TQ_A\right]\\
    &= n\operatorname{Tr}\!\left[\overline{\sigma}^{(m)}_AQ^T_A\right],
\end{align}
where
\begin{equation}
    \overline{\sigma}^{(m)}_A \coloneqq \frac{1}{n}\sum_{i=0}^{n-1}\operatorname{Tr}_{R_{i+1}}\!\left[\sigma^{i,m}_{AR_{i+1}}\right]
\end{equation}
is the average-input state to read the message $m$ encoded in the channels. Note that the average-energy constraint implies that $\operatorname{Tr}\!\left[H_A\overline{\sigma}^{(m)}_A\right]\le E$ for every $m\in \mathcal{M}_n$. Therefore, for every $n\in \mathbb{N}$, every $m\in \mathcal{M}_n$, and every operator $Q$ that satisfies the inequality in~\eqref{eq:ub_BS_op_marg}, the following inequality holds:
\begin{equation}
     \frac{1}{n}\widehat{D}\!\left(\sigma^{(m)}_{\hat{M}}\middle\Vert\tau_{\hat{M}}\right) \le \sup_{\sigma\ge 0}\left\{\operatorname{Tr}\!\left[\sigma Q^T\right]: \operatorname{Tr}\!\left[\sigma\right] = 1, \operatorname{Tr}\!\left[H\sigma\right]\le E\right\}.\label{eq:}
\end{equation}
The right-hand side of the above inequality is a semidefinite program whose dual can be found using the Lagrange multiplier method as follows:
\begin{align}
    & \sup_{\sigma\ge 0}\left\{\operatorname{Tr}\!\left[\sigma Q^T\right]: \operatorname{Tr}\!\left[\sigma\right] = 1, \operatorname{Tr}\!\left[H\sigma\right]\le E\right\} \notag \\
    &= \sup_{\sigma \ge 0} \operatorname{Tr}\!\left[\sigma Q^T\right] + \inf_{\substack{\lambda \in \mathbb{R},\\ \mu\ge 0}}\lambda\!\left(1-\operatorname{Tr}\!\left[\sigma\right]\right) + \mu\!\left(E - \operatorname{Tr}\!\left[H\sigma\right]\right)\\
    &\le \inf_{\substack{\lambda \in \mathbb{R},\\ \mu\ge 0}}\sup_{\sigma \ge 0}\operatorname{Tr}\!\left[\sigma Q^T\right] + \lambda\!\left(1-\operatorname{Tr}\!\left[\sigma\right]\right) + \mu\!\left(E - \operatorname{Tr}\!\left[H\sigma\right]\right)\\
    &= \inf_{\substack{\lambda \in \mathbb{R},\\ \mu\ge 0}} \lambda + \mu E + \sup_{\sigma \ge 0} \operatorname{Tr}\!\left[\sigma\!\left(Q^T - \lambda I - \mu H\right)\right]\\
    &=\inf_{\lambda \in \mathbb{R},\mu\ge 0}\left\{\lambda + \mu E: Q^T \le \lambda I + \mu H\right\},\label{eq:ec_op_norm_dual}
\end{align}
where the inequality follows from the max-min inequality. As such,
\begin{equation}\label{eq:BS_ub_SDP_form}
    \frac{1}{n}\widehat{D}\!\left(\sigma^{(m)}_{\hat{M}}\middle\Vert\tau_{\hat{M}}\right) \le \inf_{\lambda \in \mathbb{R},\mu\ge 0}\left\{\lambda + \mu E: Q^T \le \lambda I + \mu H\right\} \qquad \forall n\in \mathbb{N},m\in \mathcal{M}_n,
\end{equation}
and consequently,
\begin{equation}\label{eq:Umegaki_ub_SDP_form}
    \frac{1}{n}D\!\left(\sigma^{(m)}_{\hat{M}}\middle\Vert\tau_{\hat{M}}\right) \le \inf_{\lambda \in \mathbb{R},\mu\ge 0}\left\{\lambda + \mu E: Q^T \le \lambda I + \mu H\right\} \qquad \forall n\in \mathbb{N},m\in \mathcal{M}_n.
\end{equation}
Since the transpose map is a positivity preserving map,
\begin{equation}
    \operatorname{Tr}_B\!\left[\widehat{D}_{\operatorname{op}}\!\left(\Gamma^{x}_{AB},\Gamma^{\mathcal{S}}_{AB}\right)\right] \le Q_A \iff \left(\operatorname{Tr}_B\!\left[\widehat{D}_{\operatorname{op}}\!\left(\Gamma^{x}_{AB},\Gamma^{\mathcal{S}}_{AB}\right)\right]\right)^T \le Q^T_A  .
\end{equation}
We can substitute~\eqref{eq:Umegaki_ub_SDP_form} into~\eqref{eq:n_mess_umegaki_ub} and infimize over every channel $\mathcal{S}_{A\to B}$ and every $Q$ that satisfies~\eqref{eq:ub_BS_op_marg} to arrive at the following inequality, which holds for all $n\in \mathbb{N}$ and $\varepsilon\in [0,1)$:
\begin{multline}\label{eq:read_cap_opt_ub_1}
    \frac{\log_2|\mathcal{M}_n|}{n} \le \inf_{Q_A,\Gamma^{\mathcal{S}}_{AB}\ge 0}\inf_{\lambda \in \mathbb{R},\mu\ge 0}\left\{
    \begin{array}{c}
         \frac{1}{1-\varepsilon}\left(\lambda + \mu E + \frac{h_2(\varepsilon)}{n}\right): \\
         \left(\operatorname{Tr}_B\!\left[\widehat{D}_{\operatorname{op}}\!\left(\Gamma^{x}_{AB},\Gamma^{\mathcal{S}}_{AB}\right)\right]\right)^T \le Q^T_A \le \lambda I + \mu H \qquad \forall x\in \mathcal{X},\\
         \operatorname{Tr}_B\!\left[\Gamma^{\mathcal{S}}_{AB}\right] = I_A
    \end{array}
    \right\},
\end{multline}
where we have used the fact that infimizing over channels is equivalent to infimizing over Choi operators due to the Choi isomorphism. Note that we can simply eliminate $Q_A$ in~\eqref{eq:read_cap_opt_ub_1}. Also note that we do not need to include a separate support condition on the Choi operators of $\mathcal{N}^x$ and $\mathcal{S}$. This is because if $\operatorname{supp}\!\left(\Gamma^x_{AB}\right) \not\subseteq \operatorname{supp}\!\left(\Gamma^{\mathcal{S}}\right)$, the value of the objective function in~\eqref{eq:read_cap_opt_ub_1} will be $+\infty$. Therefore, a finite optimal value is achieved only when $\operatorname{supp}\!\left(\Gamma^x_{AB}\right) \subseteq \operatorname{supp}\!\left(\Gamma^{\mathcal{S}}\right)$ for every $x\in \mathcal{X}$.  Finally, acknowledging that the inequality in~\eqref{eq:read_cap_opt_ub_1} holds for every protocol $\mathcal{P}^{(n)}$, every $n\in \mathbb{N}$ and every $\varepsilon \in [0,1)$, and taking the appropriate limits in the definition of the energy-constrained reading capacity of channels given in~\eqref{eq:ec_read_cap_defn}, we conclude the statement of the theorem.
\end{proof}

\medskip

The authors of~\cite{FS19} showed that the operator inequality $\widehat{D}_{\operatorname{op}}(X,Y) \le Q$ can be approximated by a set of semidefinite conditions on $X$, $Y$, and $Q$. This means that the right-hand side of~\eqref{eq:read_cap_opt_ub_1}, which is equal to the right-hand side of~\eqref{eq:ec_read_cap_ub}, can be approximated using a semidefinite program. Thus, we conclude that the upper bound on the energy-constrained reading capacity stated in Theorem~\ref{theo:ec_read_cap_ub} is efficiently computable. We present the explicit semidefinite program in Appendix~\ref{app:UB_SDP}.

\begin{remark}
    We switch from the Umegaki relative entropy to the Belavkin--Staszewski relative entropy in~\eqref{eq:Umegaki_to_BS} to make use of the chain rule given in~\eqref{eq:ec_chain_rule_ing}. The Umegaki relative entropy also has a chain rule~\cite{FFRS20}, but it involves the regularized Umegaki relative entropy of channels, which is not known to be efficiently computable in general. However, for some special classes of channels, e.g.~the class of jointly classical-quantum channels, the regularization can be dropped. In this case, one can obtain an upper bound on the energy-constrained reading capacity in terms of the Umegaki relative entropy itself, which can be significantly tighter than the bound presented in Theorem~\ref{theo:ec_read_cap_ub}. We discuss the case of classical-quantum channels later in~\ref{sec:cq_channels}.
\end{remark}

In the absence of energy constraints, the upper bound in Theorem~\ref{theo:ec_read_cap_ub} can be expressed in a much cleaner form, as the information radius of the set of channels induced by the Belavkin--Staszewski relative entropy. We state this formally in Corollary~\ref{cor:read_cap_ub} below.

\begin{corollary}\label{cor:read_cap_ub}
    The reading capacity of a set of channels $\mathcal{N}\coloneqq \left\{\mathcal{N}^x\right\}_{x\in \mathcal{X}}$ is bounded from above by the following:
    \begin{equation}
        R\!\left(\mathcal{N}\right) \le \inf_{\mathcal{S}\in \operatorname{CPTP}}\max_{x\in \mathcal{X}}\widehat{D}\!\left(\mathcal{N}^x\middle\Vert\mathcal{S}\right),
    \end{equation}
    where the infimum is over all quantum channels.
\end{corollary}
\begin{proof}
    We can rewrite~\eqref{eq:BS_rel_ent_le_avg_BS} as follows:
    \begin{align}
        \widehat{D}\!\left(\sigma^{(m)}_{\hat{M}}\middle\Vert\tau_{\hat{M}}\right) &\le \sum_{i=0}^{n-1}\operatorname{Tr}\!\left[\left(\operatorname{Tr}_{R_{i+1}}\!\left[\sigma^{i,m}_{AR_{i+1}}\right]\right)^T\widehat{D}_{\operatorname{op}}\!\left(\Gamma^{x_{i+1}(m)}_{AB},\Gamma^{\mathcal{S}}_{AB}\right)\right] \\
        &\le \sum_{i=0}^{n-1}\left\Vert\operatorname{Tr}_B\!\left[\widehat{D}_{\operatorname{op}}\!\left(\Gamma^{x_{i+1}(m)}_{AB},\Gamma^{\mathcal{S}}_{AB}\right)\right]\right\Vert_{\infty}\\
        &= \sum_{i=0}^{n-1}\widehat{D}\!\left(\mathcal{N}^{x_{i+1}}_{A\to B}\middle\Vert\mathcal{S}_{A\to B}\right)\\
        &\le \max_{x\in \mathcal{X}}n\widehat{D}\!\left(\mathcal{N}^x_{A\to B}\middle\Vert\mathcal{S}_{A\to B}\right),
    \end{align}
    where the second inequality follows from the fact that $\left(\operatorname{Tr}_{R_{i+1}}\!\left[\sigma^{i,m}_{AR_{i+1}}\right]\right)^T$ is a unit-trace positive semidefinite operator and the equality follows from~\eqref{eq:BS_ch_closed_form}. We know from~\eqref{eq:n_mess_umegaki_ub}, combined with~\eqref{eq:BS_ge_Umegaki}, that
    \begin{equation}
    \max_{m\in \mathcal{M}_n}\frac{1}{n(1-\varepsilon)}\!\left(\widehat{D}\!\left(\sigma^{(m)}_{\hat{M}}\middle\Vert\tau_{\hat{M}}\right) + h_2(\varepsilon)\right)   
    \end{equation}
    is an upper bound on the reading rate for a protocol that generates the state $\sigma^{(m)}_{\hat{M}}$ corresponding to the message $m\in \mathcal{M}_n$ and any worst-case error tolerance $\varepsilon \in [0,1)$. Therefore, for every reading protocol $\mathcal{P}^{(n)}$, number of channel uses $n \in \mathbb{N}$, and worst-case error-tolerance $\varepsilon \in [0,1)$, the reading rate from the set of channels $\mathcal{N}$ is bounded from above by the following:
    \begin{equation}
        \frac{\log_2|\mathcal{M}_n|}{n} \le \max_{x\in \mathcal{X}}\frac{1}{1-\varepsilon}\left(\widehat{D}\!\left(\mathcal{N}^x_{A\to B}\middle\Vert\mathcal{S}_{A\to B}\right)+\frac{h_2(\varepsilon)}{n}\right)
    \end{equation}
    for every channel $\mathcal{S}_{A\to B}$. The statement of the corollary then immediately follows from the definition of reading capacity.
\end{proof}

\medskip

\subsection{Lower bound on the energy-constrained reading capacity}

Now let us turn our attention to bounding the energy-constrained reading capacity of a set of channels from below. The authors of~\cite{DW19} showed that 
\begin{equation}\label{eq:DW_lb}
    R_{H,E}\!\left(\left\{\mathcal{N}^x_{A\to B}\right\}_{x\in \mathcal{X}}\right) \ge \sup_{p(x)}\sup_{\substack{\psi_{RA},\\ \operatorname{Tr}[H_A\psi_A] \le E}}I(X;RB)_{\omega},
\end{equation}
where $I(X:RB)_{\omega} \coloneqq S(\omega_X) + S(\omega_{RB}) - S(\omega_{XRB})$ is the mutual information of the state, $S(\rho) \coloneqq -\operatorname{Tr}[\rho\log_2\rho]$ is the von Neumann entropy, and
\begin{equation}
    \omega_{XRB} \coloneqq \sum_{x\in \mathcal{X}}p(x)|x\rangle\!\langle x|_X\otimes \mathcal{N}^x_{A\to B}\!\left(\psi_{RA}\right).
\end{equation}
The supremum in~\eqref{eq:DW_lb} is over every probability distribution $p(x)$ and every pure state $\psi_{RA}$ with energy less than or equal to $E$. This rate is achieved by a non-adaptive strategy, where each memory cell is probed by the same state $\psi_{RA}$.

While the right-hand side  of~\eqref{eq:DW_lb} is not efficiently computable, the inequality in~\eqref{eq:DW_lb} can be used to obtain a weaker but efficiently computable lower bound on the energy-constrained reading capacity. One can simply pick a probability distribution $p(x)$ and an input state $\psi_{RA}$, and the resulting mutual information $I(X;RB)_{\omega}$ is a lower bound on the energy-constrained reading capacity of the set of channels $\left\{\mathcal{N}^x_{A\to B}\right\}_{x\in \mathcal{X}}$.

Here, we simply rewrite the bound in terms of the Umegaki relative entropy of states. The Umegaki relative entropy of states can be approximated by a semidefinite program, as shown in~\cite{KS26}, which in turn allows us to obtain a bilinear programming lower bound on the energy-constrained reading capacity of a set of channels. The bilinear program is not efficiently computable, but similar to the lower bound from~\cite{DW19}, any feasible point of the bilinear program yields a valid lower bound on the energy-constrained reading capacity. The advantage of rewriting the bound in this way is that one can use algorithms like the see-saw algorithm to estimate a reasonable, sometimes tight, approximation of the value of the bilinear program, excluding the need to pick a ``good'' $p(x)$ and $\psi_{RA}$ in~\eqref{eq:DW_lb}. 

Let us first present the lower bound on the energy-constrained reading capacity as an optimization involving the Umegaki relative entropy of states.

\begin{proposition}\label{theo:read_cap_lb}
    The reading capacity of a set of channels $\mathcal{N}\coloneqq \left\{\mathcal{N}^x_{A\to B}\right\}_{x\in \mathcal{X}}$ is bounded from below by the following:
    \begin{equation}\label{eq:read_cap_lb}
        R_{H,E}\!\left(\mathcal{N}\right) \ge \sup_{\left\{p(x)\right\}_{x\in \mathcal{X}}}\sup_{\substack{\psi_{RA},\\ \operatorname{Tr}[H_A\psi_A] \le E}} \sum_{x\in \mathcal{X}}p(x)D\!\left(\mathcal{N}^x_{A\to B}\!\left(\psi_{RA}\right)\middle\Vert\overline{\mathcal{N}}_{A\to B}\!\left(\psi_{RA}\right)\right),
    \end{equation}
    where $\overline{\mathcal{N}}_{A\to B} \coloneqq \sum_{x\in \mathcal{X}}p(x)\mathcal{N}^x_{A\to B}$ the maximization is over all probability distributions over the set $\mathcal{X}$ and all pure states $\psi_{RA}$ satisfying $|R| = |A|$.
\end{proposition}

\begin{proof}
    Consider a non-adaptive reading protocol for a set of channels $\mathcal{N}\coloneqq \left\{\mathcal{N}^x\right\}_{x\in \mathcal{X}}$, where the reader uses a pure tensor-product probe state $\psi^{\otimes n}_{RA}$ to read the message $m$ encoded as $\bigotimes_{i=1}^n\mathcal{N}^{x_i(m)}_{A\to B}$, followed by a joint measurement. The probe state satisfies the energy constraint $\operatorname{Tr}\!\left[H_A\psi_A\right]\le E$. In principle, this is a protocol for transmitting classical information to a receiver using $n$ signal states from the set $\left\{\mathcal{N}^{x}_{A\to B}\!\left(\psi_{RA}\right)\right\}_x$. It has been shown in~\cite{SW97,Hol98} that the maximum rate at which classical information can be transmitted using a set of signal states $\left\{\rho^x\right\}_{x\in \mathcal{X}}$ reliably, that is the error vanishes with asymptotically large number of signals, with a priori probability distribution $\left\{p(x)\right\}_{x\in \mathcal{X}}$, is equal to the following:
    \begin{equation}
        C\!\left(\left\{p(x)\right\}_{x\in \mathcal{X}}, \left\{\rho^x\right\}_{x\in \mathcal{X}}\right) = S\!\left(\sum_{x\in \mathcal{X}}p(x)\rho^x\right) - \sum_{x\in \mathcal{X}}p(x)S\!\left(\rho^x\right).
    \end{equation}
    Note that the above expression can be written as follows~\cite[Proposition~7.14]{KW20}:
    \begin{equation}
        C\!\left(\left\{p(x)\right\}_{x\in \mathcal{X}}, \left\{\rho^x\right\}_{x\in \mathcal{X}}\right) = \sum_{x\in \mathcal{X}}p(x)D\!\left(\rho^x\Vert\overline{\rho}\right),
    \end{equation}
    where $\overline{\rho}\coloneqq \sum_{x\in \mathcal{X}}p(x)\rho^x$ and $D(\cdot\Vert\cdot)$ is the Umegaki relative entropy. It is then straightforward to see that if Alice encodes a message using the set of $\mathcal{N}$ with \emph{a priori} probability distribution $\left\{p(x)\right\}_{x\in \mathcal{X}}$, then the reading capacity of $\mathcal{N}$ is bounded from below by the following:
    \begin{equation}
        R_{H,E}\!\left(\mathcal{N}\right) \ge C\!\left(\left\{p(x)\right\}_{x\in \mathcal{X}}, \left\{\mathcal{N}^x_{A\to B}\!\left(\psi_{RA}\right)\right\}_{x\in \mathcal{X}}\right)
    \end{equation}
    Since the above inequality holds for every probability distribution $\left\{p(x)\right\}_{x\in \mathcal{X}}$ and every pure state $\psi_{RA}$ that satisfies $\operatorname{Tr}\!\left[H_A\psi_A\right]\le E$, we can conclude that
    \begin{align}
        R_{H,E}\!\left(\mathcal{N}\right) &\ge \sup_{\left\{p(x)\right\}_{x\in \mathcal{X}}}\sup_{\substack{\psi_{RA},\\ \operatorname{Tr}[H_A\psi_A]\le E}}C\!\left(\left\{p(x)\right\}_{x\in \mathcal{X}}, \left\{\mathcal{N}^x_{A\to B}\!\left(\psi_{RA}\right)\right\}_{x\in \mathcal{X}}\right)\\
        &= \sup_{\left\{p(x)\right\}_{x\in \mathcal{X}}}\sup_{\substack{\psi_{RA},\\ \operatorname{Tr}[H_A\psi_A]\le E}}\sum_{x\in \mathcal{X}}p(x)D\!\left(\mathcal{N}^x_{A\to B}(\psi_{RA})\middle\Vert\sum_{x\in \mathcal{X}}p(x)\mathcal{N}^x_{A\to B}(\psi_{RA})\right)\\
        &= \sup_{\left\{p(x)\right\}_{x\in \mathcal{X}}}\sup_{\substack{\psi_{RA},\\ \operatorname{Tr}[H_A\psi_A]\le E}}\sum_{x\in \mathcal{X}}p(x)D\!\left(\mathcal{N}^x_{A\to B}(\psi_{RA})\middle\Vert\overline{\mathcal{N}}_{A\to B}(\psi_{RA})\right),\label{eq:read_cap_lb_avg_ch}
    \end{align}
    where $\overline{\mathcal{N}}_{A\to B} \coloneqq \sum_{x\in \mathcal{X}}p(x)\mathcal{N}^x_{A\to B}$ and the supremum is over every probability distribution $\left\{p(x)\right\}_{x\in \mathcal{X}}$ and pure state $\psi_{RA}$ satisfying $\operatorname{Tr}\!\left[H_A\psi_A\right]\le E$. Since the Umegaki relative entropy is invariant under isometries and every purification of $\psi_A$ is related to the canonical purification via an isometry on the reference system, we can fix $|R| = |A|$, arriving at the statement of the proposition.
    \end{proof}

    \medskip

    Now we will relax the lower bound in Proposition~\ref{theo:read_cap_lb} to a bilinear program. First, consider the following integral representation of the Umegaki relative entropy discovered in~\cite{Jen24}:
    \begin{equation}\label{eq:rel_ent_int_rep}
        D(\rho\Vert\sigma) = \frac{1}{\ln 2}\left[\int_{\mu}^{\lambda}\frac{ds}{s}\operatorname{Tr}_{+}\!\left[\sigma s - \rho\right] + \ln \lambda + 1 - \lambda\right],
     \end{equation}
     which holds for every pair of states $\rho$ and $\sigma$ if $\mu \sigma \le \rho \le \lambda \sigma$. In~\eqref{eq:rel_ent_int_rep}, $\operatorname{Tr}_{+}[\tau]$ refers to the sum of the positive eigenvalues of the operator $\tau$. The authors of~\cite{KS26}  devised a hierarchy of semidefinite programs that estimate the Umegaki relative entropy using a discretized version of~\eqref{eq:rel_ent_int_rep} and the fact that $\operatorname{Tr}_+[\tau] = \sup_{0\le P\le I}\operatorname{Tr}[P\tau]$. Their results were extended to obtain semidefinite programming hierarchies for the Umegaki relative entropy of channels in~\cite{KW26}. Using their approach, one can obtain a bilinear programming lower bound on the energy-constrained reading capacity of a finite set of channels. We report this program in Appendix~\ref{app:LB_BP}.

\subsection{Reading capacity of jointly classical-quantum channels}\label{sec:cq_channels}

In~\cite{Pascual26,PW26}, the authors obtained a single-letter expression for the reading capacity of a set of jointly classical-quantum channels. Furthermore, they showed that the reading capacity of a set of jointly classical-quantum channels can be achieved by a non-adaptive strategy. In this section, we show that the same conclusion can be reached by using our techniques.

A classical-quantum channel is defined as a channel that acts on an arbitrary input state $\rho_{RA}$ as follows:
\begin{equation}\label{eq:classical_quantum}
    \mathcal{N}_{A\to B}\!\left(\rho_{RA}\right) = \sum_{y\in \mathcal{Y}}\langle y|_A\rho_{RA}|y\rangle_A\otimes \tau^y_B,
\end{equation}
where $\left\{\tau^y\right\}_{y\in \mathcal{Y}}$ is a set of states. A set of channels are said to be jointly classical-quantum, if their action on an arbitrary state is given by~\eqref{eq:classical_quantum} with respect to a common eigenbasis $\left\{|y\rangle>\right\}_{y\in \mathcal{Y}}$. That is, $\left\{\mathcal{N}^x_{A\to B}\right\}_{x\in \mathcal{X}}$ is a set of jointly classical-quantum states if their action on an arbitrary state $\rho_{AB}$ can be expressed as follows:
\begin{equation}
    \mathcal{N}^x_{A\to B}\!\left(\rho_{RA}\right) = \sum_{y\in \mathcal{Y}}\langle y|_A\rho_{RA}|y\rangle_A\otimes \tau^{x,y}_{B},
\end{equation}
where $\tau^{x,y}_B$ is a state for every $x\in \mathcal{X},y\in \mathcal{Y}$.

It was shown in~\cite[Lemma 25]{WBHK20} that for a pair of jointly classical-quantum channels, the Umegaki relative entropy also follows a chain rule similar to the chain rule for the Belavkin--Staszewski relative entropy stated in~\eqref{eq:ec_chain_rule_ing}. In particular, let $\mathcal{N}$ and $\mathcal{M}$ be two classical-quantum channels that transform the classical label $y$ into states $\tau^y$ and $\omega^y$, respectively. Then, 
\begin{equation}\label{eq:Umegaki_chain_rule}
    D\!\left(\mathcal{N}_{A\to B}\!\left(\rho_{RA}\right)\middle\Vert\mathcal{M}_{A\to B}\!\left(\sigma_{RA}\right)\right) \le D\!\left(\rho_{RA}\middle\Vert\sigma_{RA}\right) + \sum_{y\in \mathcal{Y}}\operatorname{Tr}\!\left[\langle y|_A\rho_{RA}|y\rangle_A\right]D\!\left(\tau^y_B\middle\Vert\omega^y_B\right).
\end{equation}
This allows us to obtain a tighter upper bound on the energy-constrained reading capacity of channels compared to the bound in Theorem~\ref{theo:ec_read_cap_ub}, since we do not need to switch from the Umegaki relative entropy to the Belavkin--Staszewski relative entropy in order to make use of the chain rule. We state this new upper bound on the energy-constrained reading capacity for this in Proposition~\ref{prop:cq_channel_read_cap_ub} below.

\begin{proposition}\label{prop:cq_channel_read_cap_ub}
    Let $\mathcal{N}\coloneqq \left\{\mathcal{N}^x_{A\to B}\right\}_{x\in \mathcal{X}}$ be a set of jointly classical-quantum channels such that
    \begin{equation}
        \mathcal{N}^x_{A\to B}\!\left(\cdot\right) = \sum_{y\in \mathcal{Y}}\langle y|_A(\cdot)|y\rangle_A\otimes \omega^{x,y}_B.
    \end{equation}
    Then the energy-constrained reading capacity of $\mathcal{N}$ is bounded from above by the following:
    \begin{equation}\label{eq:ec_read_cap_cq_ub}
        R_{H,E}\!\left(\mathcal{N}\right) \le \inf_{\mathcal{S}_{A\to B}}\sup_{\substack{\rho_{RA},\\ \operatorname{Tr}[H_A\rho_A] \le E}}\sum_{y\in \mathcal{Y}} p(y)\max_{x\in \mathcal{X}}\left\{D\!\left(\omega^{x,y}_B\middle\Vert\mathcal{S}_{A\to B}(|y\rangle\!\langle y|_A)\right)\right\},
    \end{equation}
    where $p(y) \coloneqq \operatorname{Tr}\!\left[\langle y|_A\rho_{RA}|y\rangle_A\right]$, the infimum is taken over every channel $\mathcal{S}$ that is jointly classical-quantum with the channels $\left\{\mathcal{N}^x\right\}$, and the supremum is taken over all bipartite states with the energy constraint $\operatorname{Tr}\!\left[H_A\rho_A\right]\le E$.
\end{proposition}
\begin{proof}
    Let us refer back to the proof of Theorem~\ref{theo:ec_read_cap_ub} and follow all the steps up to~\eqref{eq:Before_umegaki_to_BS}, with the additional assumption that the channel $\mathcal{S}_{A\to B}$ is also a classical-quantum channel with respect to the basis $\left\{|y_A\rangle\right\}_{y\in \mathcal{Y}}$. Applying~\eqref{eq:Umegaki_chain_rule} to~\eqref{eq:Before_umegaki_to_BS}, we arrive at the following inequality:
    \begin{align}
        D\!\left(\sigma^{(m)}_{\hat{M}}\middle\Vert\tau_{\hat{M}}\right) &\le D\!\left(\sigma^{n-1,m}_{AR_n}\middle\Vert\tau^{n-1}_{AR_n}\right) + \sum_{y\in \mathcal{Y}}\operatorname{Tr}\!\left[\langle y|_A\sigma^{n-1,m}_{AR_n}|y\rangle_A\right]D\!\left(\omega^{x_n(m),y}_B\middle\Vert\mathcal{S}(|y\rangle\!\langle y|_A)\right)\\
        &\le D\!\left(\sigma^{n-1,m}_{AR_n}\middle\Vert\tau^{n-1}_{AR_n}\right) + \sum_{y\in \mathcal{Y}}\operatorname{Tr}\!\left[\langle y|_A\sigma^{n-1,m}_{AR_n}|y\rangle_A\right]\max_{x\in\mathcal{X}}\left\{D\!\left(\omega^{x,y}_B\middle\Vert\mathcal{S}(|y\rangle\!\langle y|_A)\right)\right\}\\
        &\le \sum_{y\in \mathcal{Y}}\langle y|_A\!\left(\sum_{i=0}^{n-1}\operatorname{Tr}_{R_{i+1}}\!\left[\sigma^{i,m}_{AR_{i+1}}\right]\right)|y\rangle_A\max_{x\in \mathcal{X}}\left\{D\!\left(\omega^{x,y}_B\middle\Vert\mathcal{S}(|y\rangle\!\langle y|_A)\right)\right\},
    \end{align}
    where the last inequality follows from recursively using the chain rule for the Umegaki relative entropy for jointly classical-quantum channels.

    Let us define 
    \begin{equation}
        \overline{\sigma}^{(m)}_{A} \coloneqq \frac{1}{n}\sum_{i=0}^{n-1}\sigma^{i,m}_{A}.
    \end{equation}
    Then, the following inequality holds for all $n\in \mathbb{N}$ and every message $m\in\mathcal{M}_n$:
    \begin{equation}
        \frac{1}{n}D\!\left(\sigma^{(m)}_{\hat{M}}\middle\Vert\tau_{\hat{M}}\right) \le \sum_{y\in \mathcal{Y}}\langle y|\overline{\sigma}^{(m)}_{A}|y\rangle\max_{x\in \mathcal{X}}\left\{ D\!\left(\omega^{x,y}_B\middle\Vert\mathcal{S}(|y\rangle\!\langle y|_A)\right)\right\}.
    \end{equation}
    Now, using the energy constraint $\operatorname{Tr}\!\left[H_A\overline{\sigma}^{(m)}_A\right]\le E$ and simply invoking the arguments from the proof of Theorem~\ref{theo:ec_read_cap_ub}, we conclude~\eqref{eq:ec_read_cap_cq_ub}.
\end{proof}

\medskip

In Proposition~\ref{prop:cq_channel_read_cap_ub}, the energy constraint determines which probability distributions $p(y)$ are feasible. In the absence of the energy constraint, all probability distributions are feasible, and hence, the supremum over states is equivalent to a supremum over every probability distribution on the alphabet $\mathcal{Y}$. This leads to an information-radius type bound, similar to the bound in Corollary~\ref{cor:read_cap_ub}, which we state in Corollary~\ref{cor:cq_read_cap_ub} below.

\begin{corollary}\label{cor:cq_read_cap_ub}
    Let $\mathcal{N}\coloneqq \left\{\mathcal{N}^x_{A\to B}\right\}_{x\in \mathcal{X}}$ be a set of jointly classical-quantum channels.
    The reading capacity of $\mathcal{N}$ is bounded from above by the following:
    \begin{equation}\label{eq:cq_read_cap_ub}
        R\!\left(\mathcal{N}\right) \le \inf_{\mathcal{S}_{A\to B}}\max_{x\in \mathcal{X}}D\!\left(\mathcal{N}^x_{A\to B}\middle\Vert\mathcal{S}_{A\to B}\right),
    \end{equation}
    where the infimum is over every channel $\mathcal{S}$ that is jointly classical-quantum with the set $\mathcal{N}$.
\end{corollary}
\begin{proof}
    As discussed in the text before Corollary~\ref{cor:cq_read_cap_ub}, in the absence of energy constraints, the upper bound on the reading capacity of a set of jointly classical-quantum channels given in Proposition~\ref{prop:cq_channel_read_cap_ub} can be written as follows:
    \begin{align}
        R\!\left(\mathcal{N}\right) &\le \inf_{\mathcal{S}_{A\to B}} \sup_{\rho_{A}} \sum_{y\in \mathcal{Y}}\langle y|_A\rho_{A}|y\rangle_A\max_{x\in \mathcal{X}}D\!\left(\mathcal{N}^x_{A\to B}(|y\rangle\!\langle y|_A)\middle\Vert\mathcal{S}_{A\to B}\!\left(|y\rangle\!\langle y|_A\right)\right)\\
        &= \inf_{\mathcal{S}_{A\to B}} \sup_{\left\{p(y)\right\}_{y\in \mathcal{Y}}} \sum_{y\in \mathcal{Y}}p(y)\max_{x\in \mathcal{X}}D\!\left(\mathcal{N}^x_{A\to B}(|y\rangle\!\langle y|_A)\middle\Vert\mathcal{S}_{A\to B}\!\left(|y\rangle\!\langle y|_A\right)\right)\\
        &= \inf_{\mathcal{S}_{A\to B}} \sup_{y\in \mathcal{Y}} \max_{x\in \mathcal{X}}D\!\left(\mathcal{N}^x_{A\to B}(|y\rangle\!\langle y|_A)\middle\Vert\mathcal{S}_{A\to B}\!\left(|y\rangle\!\langle y|_A\right)\right)\\
        &= \inf_{\mathcal{S}_{A\to B}} \max_{x\in \mathcal{X}} \sup_{y\in \mathcal{Y}} D\!\left(\mathcal{N}^x_{A\to B}(|y\rangle\!\langle y|_A)\middle\Vert\mathcal{S}_{A\to B}\!\left(|y\rangle\!\langle y|_A\right)\right)\\
        &= \inf_{\mathcal{S}_{A\to B}} \max_{x\in \mathcal{X}} D\!\left(\mathcal{N}^x_{A\to B}\middle\Vert\mathcal{S}_{A\to B}\right),
    \end{align}
    where the first equality follows from the fact that  $p(y) \coloneqq \langle y\rho_A|y\rangle$ is a probability distribution for every state $\rho_A$ and the final equality holds because the supremum over states in the definition of Umegaki relative entropy of channels can be replaced with the supremum over the basis states for jointly classical-quantum channels, as can be seen from~\cite[Lemma 25]{WBHK20}. 
\end{proof}

\medskip

In the absence of energy constraints, the lower bound on the reading capacity given in Proposition~\ref{theo:read_cap_lb} can also be written as the information radius bound obtained in Corollary~\ref{cor:cq_read_cap_ub} for jointly classical-quantum channels. This implies that the upper bound on the reading capacity of jointly classical-quantum channels given in Corollary~\ref{cor:cq_read_cap_ub} is indeed tight. We state this result formally in Theorem~\ref{theo:cq_read_cap_eq} below.

\begin{theorem}\label{theo:cq_read_cap_eq}
    Let $\mathcal{N}\coloneqq \left\{\mathcal{N}^x_{A\to B}\right\}_{x\in \mathcal{X}}$ be a set of jointly classical-quantum channels. Then, the reading capacity of $\mathcal{N}$ is given by the following:
    \begin{equation}\label{eq:cq_read_cap_eq}
        R\!\left(\mathcal{N}\right) = \inf_{\mathcal{S}_{A\to B}}\max_{x\in \mathcal{X}}D\!\left(\mathcal{N}^x_{A\to B}\middle\Vert\mathcal{S}_{A\to B}\right),
    \end{equation}
    where the infimum is over every channel $\mathcal{S}_{A\to B}$ that is jointly classical-quantum with $\mathcal{N}$. Moreover, this rate is achieved by a non-adaptive protocol.
\end{theorem}
\begin{proof}
    We only need to prove the achievability part because the converse was shown in Corollary~\ref{cor:cq_read_cap_ub}.

    First, note that the following equality holds for a fixed alphabet $\mathcal{X}$, every set of channels $\left\{\left\{\mathcal{N}^x\right\}_{x\in \mathcal{X}},\mathcal{S}\right\}$, every probability distribution $\left\{p(x)\right\}_{x\in \mathcal{X}}$, and every state $\psi$:
    \begin{align}
        &\sum_{x\in \mathcal{X}}p(x)D\!\left(\mathcal{N}^x(\psi)\middle\Vert\mathcal{S}(\psi)\right)\notag\\ &= -\sum_{x\in \mathcal{X}}p(x)\!\left[S\!\left(\mathcal{N}^x(\psi)\right) + \operatorname{Tr}\!\left[\mathcal{N}^x(\psi)\log_2\!\left(\mathcal{S}(\psi)\right)\right]\right]\\
        &= -\sum_{x\in \mathcal{X}}p(x)S\!\left(\mathcal{N}^x(\psi)\right) - \operatorname{Tr}\!\left[\overline{\mathcal{N}}(\psi)\log_2\!\left(\mathcal{S}(\psi)\right)\right]\\
        &= -\sum_{x\in \mathcal{X}}p(x)S\!\left(\mathcal{N}^x(\psi)\right) - \operatorname{Tr}\!\left[\overline{\mathcal{N}}(\psi)\log_2\!\left(\mathcal{S}(\psi)\right)\right] - S\!\left(\overline{\mathcal{N}}(\psi)\right) + S\!\left(\overline{\mathcal{N}}(\psi)\right)\\
        &= D\!\left(\overline{\mathcal{N}}(\psi)\middle\Vert\mathcal{S}(\psi)\right) - \sum_{x\in \mathcal{X}}p(x)S\!\left(\mathcal{N}^x(\psi)\right) + S\!\left(\overline{\mathcal{N}}(\psi)\right)\\
        &= D\!\left(\overline{\mathcal{N}}(\psi)\middle\Vert\mathcal{S}(\psi)\right) + \sum_{x\in\mathcal{X}}p(x)D\!\left(\mathcal{N}^x(\psi)\middle\Vert\overline{\mathcal{N}}(\psi)\right),
    \end{align}
    where $\overline{\mathcal{N}} \coloneqq \sum_{x\in \mathcal{X}}p(x)\mathcal{N}^x$. As such,
    \begin{equation}
        \sum_{x\in \mathcal{X}}p(x)D\!\left(\mathcal{N}^x(\psi)\middle\Vert\mathcal{S}(\psi)\right) \ge \sum_{x\in\mathcal{X}}p(x)D\!\left(\mathcal{N}^x(\psi)\middle\Vert\overline{\mathcal{N}}(\psi)\right)
    \end{equation}
    for every channel $\mathcal{S}$, and the inequality is saturated when $\mathcal{S} = \overline{\mathcal{N}}$. Therefore,
    \begin{equation}\label{eq:donald_id_inf}
        \inf_{\mathcal{S}}\sum_{x\in \mathcal{X}}p(x)D\!\left(\mathcal{N}^x(\psi)\middle\Vert\mathcal{S}(\psi)\right) = \sum_{x\in\mathcal{X}}p(x)D\!\left(\mathcal{N}^x(\psi)\middle\Vert\overline{\mathcal{N}}(\psi)\right).
    \end{equation}
    If $\mathcal{N}^x$ are jointly classical-quantum channels, then it suffices to restrict the infimum in~\eqref{eq:donald_id_inf} to an infimum over channels that are jointly classical-quantum with $\left\{\mathcal{N}^x\right\}$. This is because the optimal channel $\mathcal{S} = \overline{\mathcal{N}}$ is a convex combination of classical-quantum channels, and hence, it itself is a classical-quantum channel. Substituting~\eqref{eq:donald_id_inf} into~\eqref{eq:read_cap_lb} and removing the energy constraints, we arrive at the following inequality:
    \begin{equation}
        R\!\left(\mathcal{N}\right) \ge \sup_{\left\{p(x)\right\}_{x\in \mathcal{X}}}\sup_{\psi_{RA}} \inf_{\mathcal{S}_{A\to B}}\sum_{x\in \mathcal{X}}p(x)D\!\left(\mathcal{N}^x_{A\to B}\!\left(\psi_{RA}\right)\middle\Vert\mathcal{S}_{A\to B}\!\left(\psi_{RA}\right)\right).
    \end{equation}
    The argument of the above optimization is linear in $p(x)$ and convex in $\mathcal{S}_{A\to B}$, where the second claim follows from the joint-convexity of Umegaki relative entropy. We can also restrict the infimum over every channel $\mathcal{S}$ to an infimum over every channel such that $\operatorname{supp}(\Gamma^{\mathcal{N}^x}_{AB}) \subseteq \operatorname{supp}(\Gamma^{\mathcal{S}}_{AB})$ for every $x\in \mathcal{X}$ to ensure that the relative entropy terms are finite. Now we can use Sion's minimax~\cite{Sion58} theorem to arrive at the following inequality:
    \begin{align}
        R(\mathcal{N}) &\ge \sup_{\psi_{RA}} \inf_{\mathcal{S}_{A\to B}}\sup_{\left\{p(x)\right\}_{x\in \mathcal{X}}}\sum_{x\in \mathcal{X}}p(x)D\!\left(\mathcal{N}^x_{A\to B}\!\left(\psi_{RA}\right)\middle\Vert\mathcal{S}_{A\to B}\!\left(\psi_{RA}\right)\right)\\
        &= \sup_{\psi_{RA}} \inf_{\mathcal{S}_{A\to B}}\max_{x\in \mathcal{X}}D\!\left(\mathcal{N}^x_{A\to B}\!\left(\psi_{RA}\right)\middle\Vert\mathcal{S}_{A\to B}\!\left(\psi_{RA}\right)\right).\label{eq:read_cap_cq_lb_sup_inf_sup}
    \end{align}
    Let 
    \begin{align}
        \mathcal{N}^x_{A\to B}(\cdot) &= \sum_{y\in \mathcal{Y}}\langle y|_A(\cdot)|y\rangle_A \otimes \omega^{x,y}_B,\\
        \mathcal{S}_{A\to B}(\cdot) &= \sum_{y\in \mathcal{Y}}\langle y|_A(\cdot)|y\rangle_A \otimes \sigma^y_B.
    \end{align}
    Then,
    \begin{align}
        R(\mathcal{N}) &\ge \sup_{\psi_{RA}}\inf_{\left\{\sigma^y_B\right\}_{y\in \mathcal{Y}}}\max_{x\in \mathcal{X}}D\!\left(\sum_{y\in \mathcal{Y}}\langle y|_A\psi_{RA}|y\rangle_A\otimes \omega^{x,y}_B\middle\Vert \sum_{y\in \mathcal{Y}}\langle y|_A\psi_{RA}|y\rangle_A\otimes \sigma^{y}_B\right)\\
        &\ge \sup_{y\in \mathcal{Y}}\inf_{\left\{\sigma^y_B\right\}_{y\in \mathcal{Y}}}\max_{x\in \mathcal{X}}D\!\left( \omega^{x,y}_B\middle\Vert  \sigma^{y}_B\right),\label{eq:cq_read_cap_lb_proof}
    \end{align}
    where the last inequality follows by restricting the supremum over all pure states $\psi_{RA}$ to supremum over states of the form $|0\rangle\!\langle 0|_R\otimes |y\rangle\!\langle y|_A$. Note that each element of the set $\left\{\sigma^{y}_B\right\}_{y\in \mathcal{Y}}$ can be optimized separately for each $y\in \mathcal{Y}$. Thus, we can interchange the supremum over the alphabet $\mathcal{Y}$ and the infimum over all sets of states $\left\{\sigma^{y}_B\right\}_{y\in \mathcal{Y}}$ in~\eqref{eq:cq_read_cap_lb_proof} to arrive at the following inequality:
    \begin{align}
        R(\mathcal{N}) &\ge \inf_{\left\{\sigma^y_B\right\}_{y\in \mathcal{Y}}}\sup_{y\in \mathcal{Y}}\max_{x\in \mathcal{X}}D\!\left( \omega^{x,y}_B\middle\Vert  \sigma^{y}_B\right)\\
        &= \inf_{\left\{\sigma^y_B\right\}_{y\in \mathcal{Y}}}\max_{x\in \mathcal{X}}\sup_{y\in \mathcal{Y}}D\!\left( \omega^{x,y}_B\middle\Vert  \sigma^{y}_B\right)
    \end{align}
    Finally, noting that $\sup_{y\in \mathcal{Y}}D\!\left( \omega^{x,y}_B\middle\Vert  \sigma^{y}_B\right) = D\!\left(\mathcal{N}^x_{A\to B}\middle\Vert\mathcal{S}_{A\to B}\right)$ (see~\cite[Lemma 25]{WBHK20}) and the fact that an infimum over the set $\left\{\sigma^{y}_B\right\}_{y\in \mathcal{Y}}$ is equivalent to an infimum over channels $\mathcal{S}_{A\to B}$ that are jointly classical-quantum with $\mathcal{N}$, we arrive at the following inequality:
    \begin{equation}
        R(\mathcal{N}) \ge \inf_{\mathcal{S}_{A\to B}}\max_{x\in \mathcal{X}}D\!\left(\mathcal{N}^x_{A\to B}\middle\Vert\mathcal{S}_{A\to B}\right).
    \end{equation}
    Combining the above inequality with the statement of Corollary~\ref{cor:cq_read_cap_ub}, we arrive at~\eqref{eq:cq_read_cap_eq}. Since we started with the lower bound in Proposition~\ref{theo:read_cap_lb}, which is achieved by a non-adaptive strategy, we conclude the statement of Theorem~\ref{theo:cq_read_cap_eq}. 
\end{proof}

\medskip

\section{Numerical simulations}\label{sec:numerics}

In this section, we demonstrate the utility of our bounds through numerical simulations.

\subsection{Qubit unitary channels}

For our first example, we consider the following qubit unitary channels described in~\cite{LP17}:
\begin{equation}\label{eq:LP_channels}
    \mathcal{N}^1_{A\to B}(\cdot) = \operatorname{id}_{A\to B}(\cdot) \qquad \mathcal{N}^2_{A\to B}(\cdot) = \begin{pmatrix}
        1 & 0\\
        0 & e^{i\theta}
    \end{pmatrix}(\cdot) \begin{pmatrix}
        1 & 0\\
        0 & e^{-i\theta}
    \end{pmatrix},
\end{equation}
where $\mathcal{N}^2_{A\to B}$ is parameterized by $\theta$. This example is of particular importance in quantum reading literature because the authors of~\cite{LP17} used this example to demonstrate that probe states entangled across channel inputs can yield a strictly higher reading rate than probe states that are separable across the input channels. In particular, they showed that the reading capacity of this ensemble when using only separable probe states in a parallel protocol, which we call the single-cell probe state protocol, is given by the following expression:
\begin{equation}\label{eq:single_cell_probe_LP}
    R^1\!\left(\mathcal{N}\right) = h_2\!\left(\sin^2\!\left(\frac{\theta}{4}\right)\right),
\end{equation}
where $R^1$ denotes that the reading capacity is computed by allowing only single-cell probe states. 

On the other hand, they showed that there exists a probe state entangled across the inputs of two channels in a parallel protocol, which we call the two-cell probe state protocol, that yields a reading rate equal to $\frac{1}{2}h_2\!\left(\sin^2\!\left(\frac{\theta}{2}\right)\right)$. As such,
\begin{equation}\label{eq:2_cell_probe_st_LP}
    R\!\left(\mathcal{N}\right) \ge \frac{1}{2}h_2\!\left(\sin^2\!\left(\frac{\theta}{2}\right)\right),
\end{equation}
which is strictly larger than $R^1(\mathcal{N})$ for a range of the parameter $\theta$.

\begin{figure}
    \centering
    \includegraphics[width=0.7\linewidth]{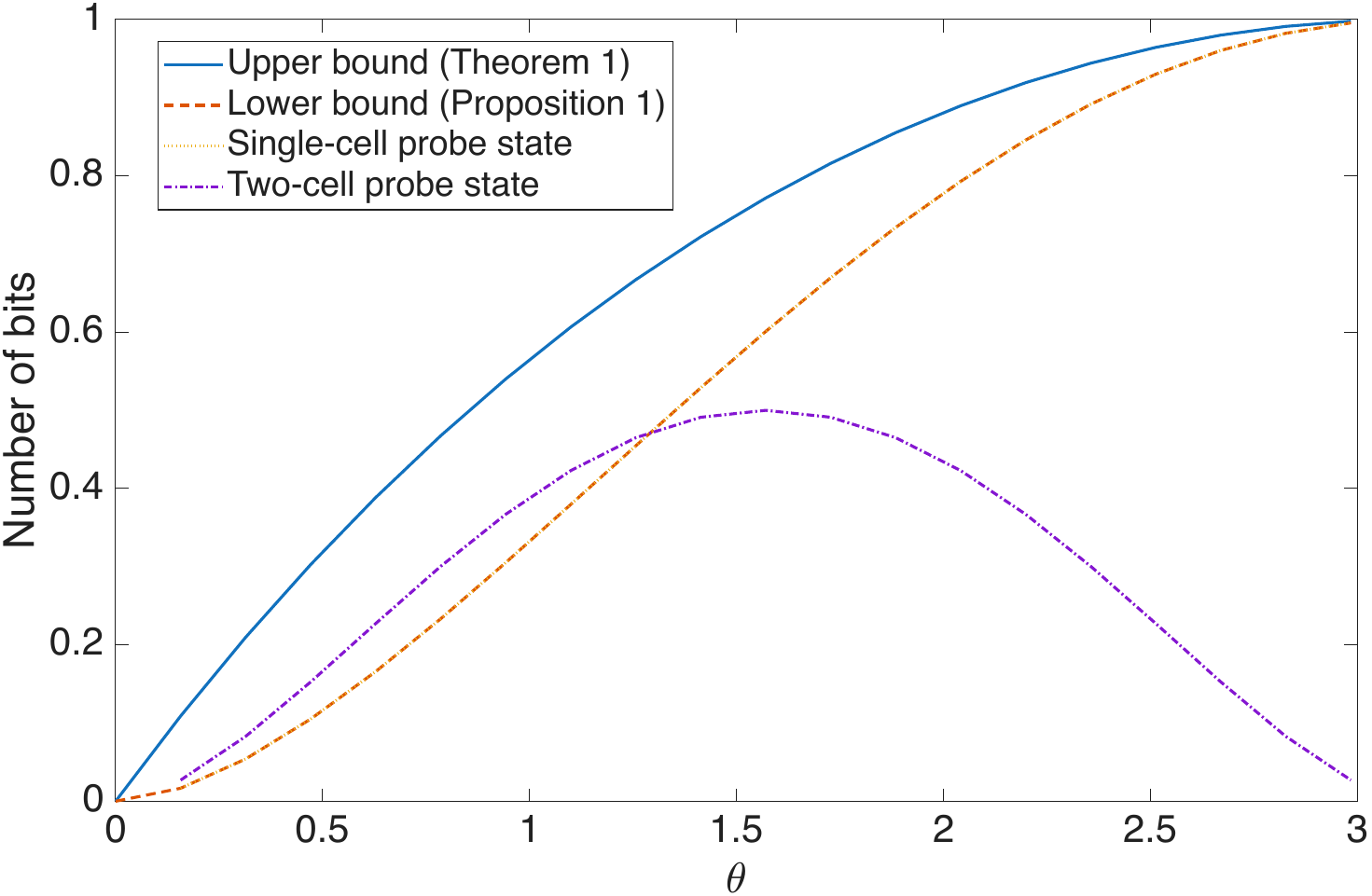}
    \caption{Comparison of upper and lower bounds on the unconstrained reading capacity of the pair of qubit unitary channels defined in~\eqref{eq:LP_channels}. The lower bound computed using the bilinear program in Appendix~\ref{app:LB_BP} overlaps with the single-cell probe state bound given in~\eqref{eq:single_cell_probe_LP}.}
    \label{fig:unconstrained_LP}
\end{figure}

In Figure~\ref{fig:unconstrained_LP}, we plot our upper bounds from Theorem~\ref{theo:ec_read_cap_ub} and lower bounds from~\Cref{theo:read_cap_lb} for the pair of channels defined in~\eqref{eq:LP_channels}. For this plot, we consider the unconstrained setting where probe states of arbitrary energy are allowed. We compare our bounds with the reading capacity of this ensemble of channels achieved by only single-cell probe states given in~\eqref{eq:single_cell_probe_LP} and the lower bound on the reading capacity of this ensemble of channels given in~\eqref{eq:2_cell_probe_st_LP}. To compute our lower bounds from~\Cref{theo:read_cap_lb}, we run the bilinear program given in Appendix~\ref{app:LB_BP} and use the state $\psi_{RA}$ and probability distribution $\left\{p(x)\right\}_{x\in \{1,2\}}$ generated by the bilinear program to compute $I(X;RB)_{\omega}$ given in~\eqref{eq:DW_lb}. We find that our lower bounds from~\Cref{theo:read_cap_lb} overlap with the single-cell probe bound given in~\eqref{eq:single_cell_probe_LP}.

\subsection{Loss--dephasing channels}

Two common sources of noise in photonic systems are the loss of photons and phase damping. These two effects are encapsulated in a single noise model by loss--dephasing channels. The channel is described by two parameters: $\eta$ dictating the transmissivity of photons, and $\gamma$ dictating the strength of phase-damping noise. For simplicity, we directly present the Choi operator of the loss--dephasing channel below, and we refer the reader to~\cite[Appendix G]{HLSW26} for the derivation of the Choi operator:
\begin{equation}
    \Gamma^{\eta,\gamma}_{RB} = \sum_{m,n=0}^{\infty}\sum_{k=0}^{\min(m,n)}
\sqrt{\binom{m} { k}\binom{n} { k}}\eta^{\frac{1}{2}(m+n)-k}(1-\eta)^{k} e^{-\frac{\gamma}{2}(m-n)^2} |m\rangle\!\langle n|_R\otimes\ket{m-k}\!\bra{n-k}_B,
\end{equation}
where $\eta \in [0,1]$ and $\gamma \in [0,\infty)$.

For numerical calculations, we truncate the Choi operator of the loss--dephasing channel to a finite dimensional operator. Such truncation is expected to induce an error in the computation of entropic quantities that involve the Choi operators. However, numerical evidence suggests that low-dimensional approximations of the Choi operator may suffice to produce ``good'' approximations of the actual values of the entropic quantities in question for small values of the energy budget $E$ (see, for example, the discussion around Fig.~5 in~\cite{HLSW26}). Moreover, we note that the purpose of the numerical simulations for loss--depahsing channels is to demonstrate the difference in the upper and lower bounds on the energy-constrained reading capacity presented in this work rather than give precise upper and lower bounds on the energy-constrained reading capacity of these bosonic channels.

\begin{figure}
    \centering
    \includegraphics[width=0.7\linewidth]{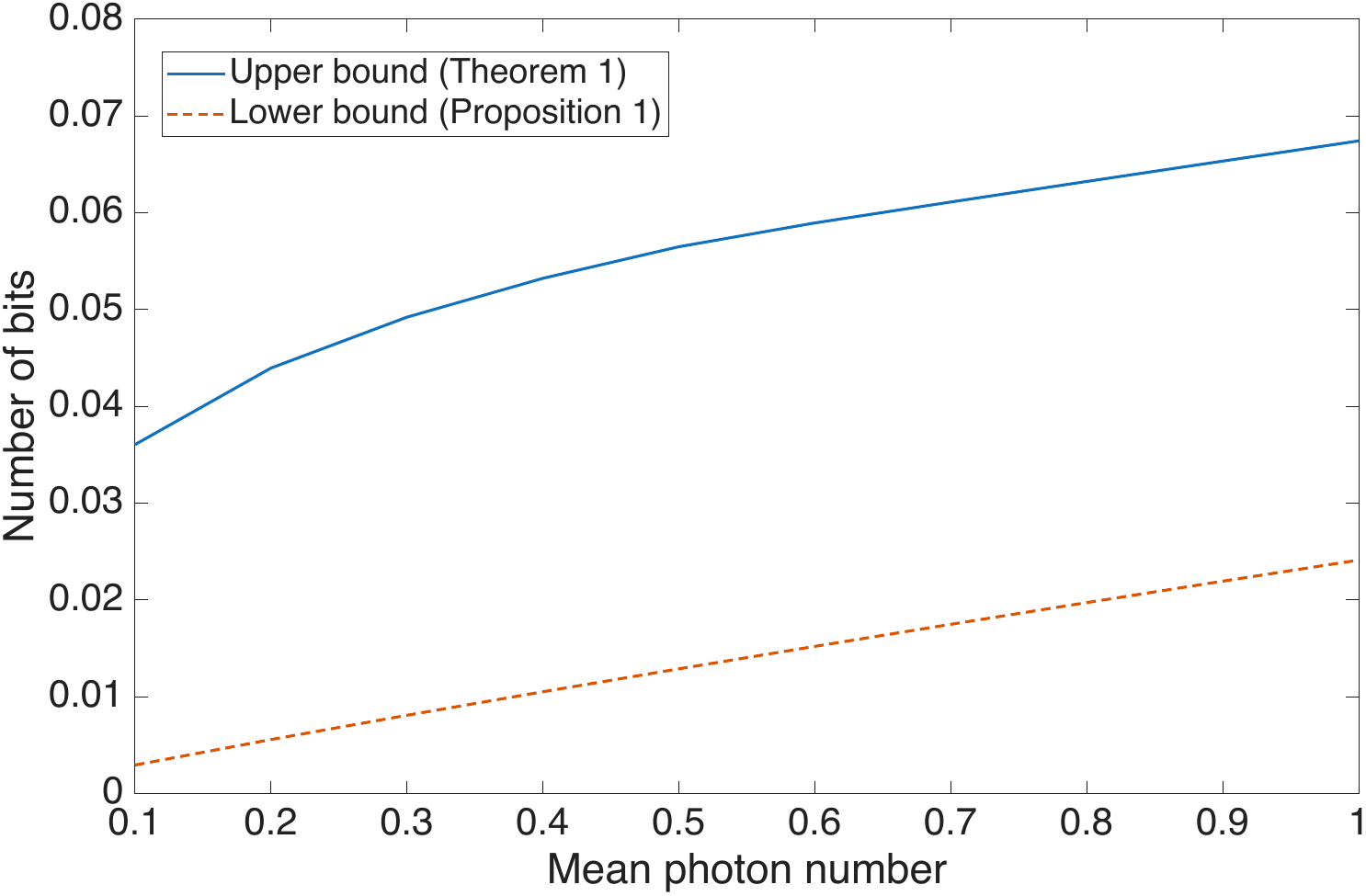}
    \caption{Upper and lower bounds on the energy constrained reading capacity of a pair of loss--dephasing channels plotted against the average photon number per channel probe. In this example, the first channel has the parameters $\eta_1 = 0.95, \gamma_1 = 0.01$, and the second channel has the parameters $\eta_2 = 0.85, \gamma_2 = 0.01$. The output dimension of the channel is truncated to four.}
    \label{fig:loss-dephasing}
\end{figure}

In Figure~\ref{fig:loss-dephasing}, we compare the upper and lower bounds on the energy-constrained reading capacity of a pair of loss--dephasing channels, with parameters $(\eta_1 = 0.95,\gamma_1 = 0.01)$ and $(\eta_2 = 0.85,\gamma_2 = 0.01)$, respectively. The Hamiltonian is chosen to be the photon number operator so that $\operatorname{Tr}\!\left[H\rho\right]$ represents the mean photon number. The upper bounds are computed using the Theorem~\ref{theo:ec_read_cap_ub} (see Appendix~\ref{app:UB_SDP} for the explicit semidefinite program). To compute the lower bounds, we use the bilinear program given in Appendix~\ref{app:LB_BP} to find ``good'' estimates of the probe state $\psi_{RA}$ and the probability distribution $\left\{p(x)\right\}_{x\in \{0,1\}}$ that are then used to evaluate the lower bound given in~\Cref{theo:read_cap_lb}.

\section{Discussion}

\textbf{Summary: }We have developed complementary upper and lower bounds on the energy-constrained reading capacity of a finite set of channels.  The upper bound can be computed efficiently using a semidefinite program. In the absence of energy constraints, the upper bound can be expressed as the information radius of the set of channels under the Belvakin--Staszewski realtive entropy. 

We express the lower bound on the energy-constrained quantum reading capacity as a bilinear program. We demonstrate via numerical simulations that this bilinear program can be used to obtain a ``good'' set of parameters to evaluate a relaxation of the lower bound obtained in~\cite{DW19}, which, on its own, is not efficiently computable with respect to the dimensions of the channels.

For jointly classical-quantum channels, our upper bounds reduce to a simpler form involving Umegaki relative entropy. In the absence of energy constraints, our upper bounds coincide with the lower bounds for jointly classical-quantum channels, establishing an exact expression for the reading capacity of a finite set of jointly classical-quantum channels as the information radius of the set of channels under the Umegaki relative entropy. Our results agree with the findings of~\cite{Pascual26,PW26} that the reading capacity of a set of jointly classical-quantum channels can be reached by non-adaptive strategies.

\textbf{Future directions: }Moving forward, there are several directions that can be investigated. In the absence of energy constraints, our upper bound can be expressed as the information radius of the set of channels used in the quantum reading task. It would be illuminating to see if upper and lower bounds in the energy-constrained setting can be expressed as the information radius of the channels with respect to some energy-constrained divergences.

Our methods allow us to demonstrate the optimality of non-adaptive protocols for quantum reading with jointly classical-quantum channels in the absence of energy constraints. However, it is unclear if the adaptive protocols can achieve a higher capacity than every non-adapative protocol for jointly classical-quantum channels in the presence of energy constraints.

It will also be interesting to find more classes of channels for which adaptive protocols are sufficient to achieve the quantum reading capacity associated with a set of such channels.

\section*{Acknowledgement}

We thank Ludovico Lami for fruitful discussions, especially regarding the state-dependent chain rule for the Belavkin--Staszewski relative entropy from \cite{HLSW26}.

We also thank V\'ictor Pascual Abraldes and Andreas Winter for sharing an early version of their work, which helped us identify that, in agreement with their results, our techniques also arrive at the conclusion that non-adaptive protocols achieve the quantum reading capacity for jointly classical-quantum channels.

ZH is supported by an ARC DECRA Fellowship (DE230100144) ``Quantum-enabled super-resolution
imaging'' and an RMIT Vice Chancellor’s Senior Research Fellowship.
MMW acknowledges support from the National Science Foundation under Grant No.~2329662.

\bibliography{Ref}
\bibliographystyle{alphaurl}

\appendix

\section{Semidefinite program for an upper bound on the energy-constrained quantum reading capacity}\label{app:UB_SDP}

In this section, we present the semidefinite program to compute the upper bound on the energy-constrained reading capacity of a set of channels given in Theorem~\ref{theo:ec_read_cap_ub}.

It was shown in~\cite{FS19} that the set of tuples $(X,Y,\Theta)$ that satisfy $\Theta \le -\widehat{D}_{\operatorname{op}}\!\left(X,Y\right)$ can be approximated by the following set of semidefinite constraints described for some positive integers $m$ and $k$:
\begin{equation}
    \label{eq:K_cone}
    \left\{\begin{array}{c}
         \Theta:  \\
         \frac{1}{\ln(2)}\sum_{j=1}^m w_jT_j = 2^{-k}\Theta, \quad Z_0 = Y,\\
         \left[\begin{array}{cc}
       X  & Z_{i+1} \\
        Z_{i+1} & Z_i
    \end{array}\right]\ge 0 \qquad \forall i\in \{0,1,\ldots,k-1\},\\
    \left[\begin{array}{cc}
        Z_k - X - T_j & -\sqrt{t_j}T_j \\
         -\sqrt{t_j}T_j & X -t_jT_j
    \end{array}\right] \ge 0 \qquad \forall j\in \{1,2,\ldots,m\}
    \end{array}\right\},
\end{equation}
where $\left\{T_j\right\}_{j=1}^{m}$ and $\left\{Z_i\right\}_{i=0}^{k}$ are sets of Hermitian operators, and $\left\{w_j\right\}_{j=1}^m$ and $\left\{t_j\right\}_{j=1}^m$ are the weights and nodes respectively, of the $m$-point Gauss-Legendre quadrature on the interval $[0,1]$. The accuracy of the approximation improves with increasing $m$ and $k$, and the approximation is tight in the limit $m,k\to \infty$. 

Substituting the semidefinite approximation of $\Theta \le -\widehat{D}_{\operatorname{op}}\!\left(X,Y\right)$ in~\eqref{eq:ec_read_cap_ub}, which is equivalent to the condition $-\Theta \ge \widehat{D}_{\operatorname{op}}\!\left(X,Y\right)$, we arrive at the following semidefinite programming approximation of the upper bound in Theorem~\ref{theo:ec_read_cap_ub}:
\begin{multline}
    R_{H,E}\!\left(\mathcal{N}\right) \lesssim\\ \inf\left\{\begin{array}{c}
         \lambda + \mu E:  \\
         -\left(\operatorname{Tr}_B\!\left[\Theta^x_{AB}\right]\right)^T \le \lambda I + \mu H \qquad \forall x\in \mathcal{X},\\
         \frac{1}{\ln(2)}\sum_{j=1}^m w_jT^x_j = 2^{-k}\Theta^x, \quad Z^x_0 = S \qquad \forall x\in \mathcal{X},\\
         \left[\begin{array}{cc}
        N^x & Z^x_{i+1} \\
        Z^x_{i+1} & Z^x_i
    \end{array}\right]\ge 0 \qquad \forall i\in \{0,1,\ldots,k-1\}, x\in \mathcal{X},\\
    \left[\begin{array}{cc}
        Z^x_k - N^x - T^x_j & -\sqrt{t_j}T^x_j \\
         -\sqrt{t_j}T^x_j & N^x -t_jT^x_j
    \end{array}\right] \ge 0 \qquad \forall j\in \{1,2,\ldots,m\}, x\in \mathcal{X},\\
    \operatorname{Tr}_B\!\left[S_{AB}\right] = I_A
    \end{array}\right\},
\end{multline}
where $N^x_{AB}$ are the Choi operators of the channels in the set $\mathcal{N}$ and the infimum is over $\lambda \in \mathbb{R},\mu\ge 0$, Hermitian operators $\left\{\Theta^x_{AB}\right\}_{x\in \mathcal{X}}$, $\left\{Z^x_i\right\}_{i\in \{0,1,\ldots,k\},x\in \mathcal{X}}$, and $\left\{T^x_j\right\}_{j\in \{1,2,\ldots,m\},x\in \mathcal{X}}$, and positive semidefinite operator $S_{AB} \ge 0$.

\section{Bilinear program for a lower bound on the energy-constrained quantum reading capacity}\label{app:LB_BP}

In this section, we obtain a bilinear programming lower bound on the quantum reading capacity using~\Cref{theo:read_cap_lb}.

In~\cite[Appendix A]{KS26}, the authors showed that for every set $\left\{t_k\right\}_{k=1}^r$ such that $\mu = t_1 \le t_2 \le \cdots \le t_r = \lambda$, the following inequality holds:
\begin{equation}
    \int_{\mu}^{\lambda}\frac{ds}{s}\operatorname{Tr}_+\!\left[\sigma s - \rho\right] \ge \sum_{k=1}^{r-1}\sup_{0\le P_k\le I}\operatorname{Tr}\!\left[P_k\!\left(\sigma \beta_k + \rho\alpha_k\right)\right],
\end{equation}
where $\alpha_k \coloneqq \ln\frac{t_k}{t_{k+1}}$ and $\beta_k \coloneqq t_{k+1} - t_k$ and $\mu \sigma \le \rho \le \lambda\sigma$.

Note that every bipartite pure state $\psi_{RA}$ can be written as $\psi_{RA}\coloneqq X_R\Gamma_{RA}X^{\dagger}_{R}$, where $X_R$ is a linear operator with $\operatorname{Tr}\!\left[X^{\dagger}_RX_R\right] = 1$, where $\Gamma_{RA} \coloneqq \sum_{i,j = 0}^{d-1}|i\rangle\!\langle j|_R\otimes |i\rangle\!\langle j|_A$ is the unnormalized maximally entangled operator of Schmidt rank $d$. Let  $\mathcal{N}_{A\to B}$ and $\mathcal{M}_{A\to B}$ be arbitrary channels with $\operatorname{supp}\!\left(\Gamma^{\mathcal{N}}\right) \subseteq \operatorname{supp}\!\left(\Gamma^{\mathcal{M}}\right)$, and let $\psi_{RA}$ be an arbitrary pure state. Then,
\begin{align}
    &\int_{\mu}^{\lambda}\frac{ds}{s}\operatorname{Tr}_+\!\left[\mathcal{M}_{A\to B}\!\left(\psi_{RA}\right) s - \mathcal{N}_{A\to B}\!\left(\psi_{RA}\right)\right]\notag\\ &\ge \sum_{k=1}^{r-1}\sup_{0\le P^k\le I}\operatorname{Tr}\!\left[P^k_{RB}\!\left(\mathcal{M}_{A\to B}\!\left(\psi_{RA}\right) \beta_k + \mathcal{N}_{A\to B}\!\left(\psi_{RA}\right)\alpha_k\right)\right]\\
    &= \sum_{k=1}^{r-1}\sup_{0\le P^k\le I}\operatorname{Tr}\!\left[P^k_{RB}\!\left(\mathcal{M}_{A\to B}\!\left(X_R\Gamma_{RA}X^{\dagger}_R\right) \beta_k + \mathcal{N}_{A\to B}\!\left(X_R\Gamma_{RA}X^{\dagger}_R\right)\alpha_k\right)\right]\\
    &= \sum_{k=1}^{r-1}\sup_{0\le P^k\le I}\operatorname{Tr}\!\left[X^{\dagger}_RP^k_{RB}X_R\!\left(\mathcal{M}_{A\to B}\!\left(\Gamma_{RA}\right) \beta_k + \mathcal{N}_{A\to B}\!\left(\Gamma_{RA}\right)\alpha_k\right)\right]\\
    &= \sum_{k=1}^{r-1}\sup_{0\le P^k\le I}\operatorname{Tr}\!\left[X^{\dagger}_RP^k_{RB}X_R\!\left(\Gamma^{\mathcal{M}}_{RB} \beta_k + \Gamma^{\mathcal{N}}_{RB}\alpha_k\right)\right],\label{eq:int_rep_w_P}
\end{align}
where the penultimate equality follows from the cyclicity of trace and the final equality follows from the definition of Choi operators of channels given in~\eqref{eq:Choi_op_defn}. Let us define $\Omega^k_{RB} \coloneqq X^{\dagger}_R P^k_{RB} X_R$. Since $0\le P^k_{RB} \le I_{RB}$, we have $0 \le \Omega^k_{RB} \le X^{\dagger}_RX_R \otimes I_B$. Now let us define $\rho \coloneqq X^{\dagger}_RX_R$, which is a unit-trace positive semidefinite operator. We can then rewrite~\eqref{eq:int_rep_w_P} as follows:
\begin{equation}
    \int_{\mu}^{\lambda}\frac{ds}{s}\operatorname{Tr}_+\!\left[\mathcal{M}_{A\to B}\!\left(\psi_{RA}\right) s - \mathcal{N}_{A\to B}\!\left(\psi_{RA}\right)\right] \ge \sum_{k=1}^{r-1}\sup_{0\le \Omega^k\le \rho}\operatorname{Tr}\!\left[\Omega^k_{RB}\!\left(\Gamma^{\mathcal{M}}_{RB} \beta_k + \Gamma^{\mathcal{N}}_{RB}\alpha_k\right)\right],\label{eq:int_rep_w_omega}
\end{equation}
where the operator $\rho_R$ is determined by $\psi_{RA}$.

Let us get back to the lower bound on the quantum reading capacity given in~\Cref{theo:read_cap_lb}. Let $N^x_{AB}$ be the Choi operator of the channel $\mathcal{N}^x_{A\to B}$ for every $x\in \mathcal{X}$. Let $\overline{N}_{AB}$ be the Choi operator of the average channel $\sum_{x\in \mathcal{X}}p(x)\mathcal{N}^x_{A\to B}$ with respect to some probability distribution $p(x)$. Then using~\eqref{eq:int_rep_w_omega} and~\eqref{eq:rel_ent_int_rep}, we can write
\begin{multline}
    D\!\left(\mathcal{N}^x_{A\to B}\!\left(\psi_{RA}\right)\middle\Vert\overline{\mathcal{N}}_{A\to B}\!\left(\psi_{RA}\right)\right) \ge \\ \frac{1}{\ln 2}\!\left[\sum_{k=1}^{r-1}\sup_{0\le \Omega^{x,k}\le \rho}\operatorname{Tr}\!\left[\Omega^{x,k}_{AB}\!\left(\overline{N}_{AB} \beta_k + N^x_{AB}\alpha_k\right)\right] + \ln \lambda + 1 - \lambda\right] \qquad \forall x\in \mathcal{X},
\end{multline}
where $\mu$ and $\lambda$ are chosen such that $\mu \overline{N} \le N^x \le \lambda \overline{N}$ for every $x\in \mathcal{X}$, and $\alpha_k$ and $\beta_k$ are defined accordingly by choosing a fixed discretization $\mu = t_1\le t_2 \le \cdots \le t_r = \lambda$ for every $x\in \mathcal{X}$.

Now, consider the following quantity:
\begin{multline}
    \sum_{x\in \mathcal{X}}p(x)D\!\left(\mathcal{N}^x_{A\to B}\!\left(\psi_{RA}\right)\middle\Vert\overline{\mathcal{N}}_{A\to B}\!\left(\psi_{RA}\right)\right) \ge \\ \frac{1}{\ln 2}\!\left[\sum_{x\in \mathcal{X}}\sum_{k=1}^{r-1}\sup_{0\le \Omega^{x,k}\le \rho}\operatorname{Tr}\!\left[p(x)\Omega^{x,k}_{AB}\!\left(\overline{N}_{AB} \beta_k + N^x_{AB}\alpha_k\right)\right] + \ln \lambda + 1 - \lambda\right].
\end{multline}
Let us define $Z^{x,k}_{AB}\coloneqq p(x)\Omega^{x,k}_{AB}$. The condition $0\le \Omega^{x,k}_{AB}\le \rho_A\otimes I_B$ translates to $0\le Z^{x,k}_{AB} \le p(x)\rho_A\otimes I_B$ for every $x\in \mathcal{X}$. As such,
 \begin{multline}
    \sum_{x\in \mathcal{X}}p(x)D\!\left(\mathcal{N}^x_{A\to B}\!\left(\psi_{RA}\right)\middle\Vert\overline{\mathcal{N}}_{A\to B}\!\left(\psi_{RA}\right)\right) \ge \\ \frac{1}{\ln 2}\!\left[\sum_{x\in \mathcal{X}}\sum_{k=1}^{r-1}\sup_{0\le Z^{x,k}\le p(x)\rho}\operatorname{Tr}\!\left[Z^{x,k}_{AB}\!\left(\overline{N}_{AB} \beta_k + N^x_{AB}\alpha_k\right)\right] + \ln \lambda + 1 - \lambda\right].
\end{multline}
Now we fix $\mu$ and $\lambda$, and substitute the above inequality in~\eqref{eq:read_cap_lb} with an additional constraint $N^x \le \lambda\overline{N}$ for every $x\in \mathcal{X}$ to ensure the Umegaki relative entropy terms are finite. It may not be possible to find a positive $\mu$ such that $\mu\overline{N} \le N^x$. However, one can start with a small positive value $\mu = \delta$ that incurs an error of at most $\frac{\delta}{\ln 2}$. We refer the reader to~\cite{KS26} for a thorough discussion of this case.

The energy constraint in~\eqref{eq:read_cap_lb} can be reformulated to be expressed as  a constraint on $\rho$ as follows: 
\begin{align}
    \operatorname{Tr}\!\left[\left(I_R\otimes H_A\right)\psi_{RA}\right] &= \operatorname{Tr}\!\left[\left(I_R\otimes H_A\right)X_R\Gamma_{RA}X^{\dagger}_R\right]\\
    &= \operatorname{Tr}\!\left[\left(I_R\otimes H_A\right)X^T_A\Gamma_{RA}X^{*}_A\right]\label{eq:transpose_trick_app}\\
    &= \operatorname{Tr}\!\left[H_A\!\left(X^{\dagger}_AX_A\right)^T\right]\label{eq:Gamma_tr_id}\\
    &= \operatorname{Tr}\!\left[H_A\rho^T_A\right],
\end{align}
where $X^*$ denotes the complex conjugate of $X$. The equality in~\eqref{eq:transpose_trick_app} follows from the transpose trick and the equality in~\eqref{eq:Gamma_tr_id} follows from the fact that the marginal of the unnormalized maximally entangled operator on either system is the identity operator. 
The inequality in~\eqref{eq:read_cap_lb} can now be written as follows:
\begin{multline}\label{eq:read_cap_lb_opt}
    R_{H,E}(\mathcal{N})\ge \frac{1}{\ln 2}\!\left(\ln \lambda + 1 - \lambda\right)\\ +\sup_{\left\{p(x)\right\}_{x\in \mathcal{X}}}\sup_{\substack{\rho\ge 0,\\ \left\{Z^{x,k}\right\}_{x\in \mathcal{X},k\in \{1,2,\ldots,r-1\}}\ge 0}}\left\{\begin{array}{c}
         \frac{1}{\ln 2}\!\left[\sum_{x\in \mathcal{X}}\sum_{k=1}^{r-1}\operatorname{Tr}\!\left[Z^{x,k}_{AB}\!\left(\overline{N}_{AB} \beta_k + N^x_{AB}\alpha_k\right)\right] \right]:\\
         \operatorname{Tr}[\rho] = 1, \operatorname{Tr}\!\left[H_A\rho^T_A\right] \le E,\\
         0\le Z^{x,k}_{AB} \le p(x)\rho_A\otimes I_B \quad \forall x\in \mathcal{X},k\in \{1,2,\ldots,r-1\},\\
         \mu\overline{N} \le N^x \le \lambda\overline{N}\quad \forall x\in \mathcal{X}
    \end{array}\right\}.
\end{multline}
Since we have fixed $\mu$ and $\lambda$, the supremum over the probability distributions may not be able to span all probability distributions without violating the condition $\mu\overline{N} \le N^x \le \lambda\overline{N}$. Nonetheless, for every fixed $\mu$ and $\lambda$, the resulting optimization yields a lower bound on the energy-constrained quantum reading capacity. It must be noted that in the case where $\operatorname{supp}\!\left(N^x_{AB}\right) = \operatorname{supp}\!\left(N^y_{AB}\right)$ for every $x,y\in \mathcal{X}$, there exist finite $\mu$ and $\lambda$ such that the condition $\mu\overline{N} \le N^x \le \lambda\overline{N}$ is satisfied for every probability distribution $p(x)$ and every $x\in\mathcal{X}$.

For a fixed probability distribution $\left\{p(x)\right\}_{x\in \mathcal{X}}$, the optimization in~\eqref{eq:read_cap_lb_opt} is a semidefinite program with variables $\rho$ and $\left\{Z^{x,k}_{AB}\right\}_{x\in \mathcal{X},k\in \{1,2,\ldots,r-1\}}$. For a fixed choice of $\rho$ and $\left\{Z^{x,k}_{AB}\right\}_{x\in \mathcal{X},k\in \{1,2,\ldots,r-1\}}$, the optimization is a semidefinite program in $\left\{p(x)\right\}_{x\in \mathcal{X}}$. The see-saw method (referred to as the alternate convex search in~\cite[Section 4.2.1]{GPK07}) can then provide an efficient way to find feasible points that yield a tighter lower bound on the energy-constrained quantum reading capacity than the ones obtained by guessing ``good'' values of $\left\{p(x)\right\}_{x\in \mathcal{X}}$ and $\psi_{RA}$ in~\eqref{eq:read_cap_lb}. Note that every feasible point of the optimization in~\eqref{eq:read_cap_lb_opt} yields a valid lower bound on the energy-constrained reading capacity of the set of channels $\mathcal{N}$. Moreover, the program also yields a value of $\rho$, which tells us the probe state one must use to achieve the reading rate obtained by the optimization in~\eqref{eq:read_cap_lb_opt}. In fact, one can use these values of input probe state and probability distribution to evaluate the Holevo information in~\Cref{theo:read_cap_lb}, which gives a tighter bound than the output of~\eqref{eq:read_cap_lb_opt} since it does not suffer from relaxations induced by discretization of the integral in~\eqref{eq:rel_ent_int_rep}.
\end{document}